\pdfoutput=1
\documentclass[letterpaper,USenglish,cleveref,autoref,thm-restate]{lipics-v2021}
\nolinenumbers
\usepackage[utf8]{inputenc}
\RequirePackage{amsfonts}
\RequirePackage{amssymb}
\RequirePackage{mathtools}
\usepackage[T1]{sansmath}
\SetMathAlphabet{\mathsfbf}{sans}{\sansmathencoding}{\sfdefault}{bx}{sl}

\usepackage{amsthm}
\usepackage{xspace}

\usepackage[table]{xcolor}
\usepackage{listings}

\definecolor{darkgray}{rgb}{.4,.4,.4}

\usepackage{array,tabularx}
\usepackage[inline]{enumitem}

\usepackage{csquotes}
\usepackage{expl3}
\usepackage{ebproof}

\newcommand\hypo{\Hypo}
\newcommand\infer{\Infer}

\usepackage[style=base]{caption}
\usepackage{subcaption}
\definecolor{lipicsgrey}{HTML}{99999c}
\hypersetup{
	colorlinks=false,
	allbordercolors=white,
	urlbordercolor=lipicsgrey,
}

\usepackage{amsmath}
\theoremstyle{remark}
\newtheorem{notation}{Notation}

\AtBeginDocument{

}

\newcommand{\BNF}{\enspace \ensuremath{\Vert} \enspace} %

\DeclareMathOperator{\id}{id}

\DeclareMathOperator{\var}{var}
\DeclareMathOperator{\Card}{Card}
\newcommand{\mat}[1]{\left(\begin{smallmatrix}#1\end{smallmatrix}\right)} %
\newcommand{\zmat}{\mathbf{0}} %
\newcommand{\umat}{\mathbf{1}} %

\newcommand{\pr}{\lstinline[mathescape]}

\ebproofnewstyle{small}{
	separation = 1em, rule margin = .5ex,
}

\newcommand*{\eg}{e.g.\@\xspace}
\newcommand*{\cf}{cf.\@\xspace}
\newcommand*{\ie}{i.e.\@\xspace}

\makeatletter
\newcommand*{\etc}{%
	\@ifnextchar{.}%
	{etc}%
	{etc.\@\xspace}%
}
\makeatother

\usepackage{soul}
\renewcommand{\st}{s.t.\@\xspace}
\newcommand*{\resp}{resp.\@\xspace}
\newcommand*{\wrt}{w.r.t.\@\xspace}
\newcommand*{\wlg}{w.l.o.g.\@\xspace}

\makeatletter
\renewcommand{\boxed}[1]{\text{\fboxsep=.2em\fbox{\m@th$\displaystyle#1$}}}
\makeatother

\newcommand{\vdashJK}{\vdash_{\textnormal{\textsc{jk}}}}
\newcommand{\mwp}{\textnormal{\textsc{mwp}}\xspace} %

\newcommand{\hrefdgfusion}{https://github.com/statycc/pymwp/blob/946a5b44692325095392694950ed03807f059b52/pymwp/delta_graphs.py\#L274}

\newcommand{\hrefpolymulti}{https://github.com/statycc/pymwp/blob/746da71a5490c5f21ebc5643ea20822f78876959/pymwp/polynomial.py\#L199}

\newcommand{\hreffeatlist}{https://statycc.github.io/pymwp/features/}

\newcommand{\hrefexplosion}{https://github.com/statycc/pymwp/blob/bded4f5fa70bcc8c4997a73f0444574a7bb5e4a1/c_files/other/explosion.c}
\newcommand{\hrefexplosiondemo}{https://statycc.github.io/pymwp/demo/\#other_explosion.c}
\newcommand{\hrefexassign}{https://statycc.github.io/pymwp/demo/\#basics_assign_expression.c}
\newcommand{\hrefeximproved}{https://statycc.github.io/pymwp/demo/\#implementation_paper_example7.c}

\newcommand{\hrefcontains}{https://github.com/statycc/pymwp/blob/b2974a26cf42c262d84662e32869b0e2c3950216/pymwp/monomial.py\#L79}
\newcommand{\hrefincludes}{https://github.com/statycc/pymwp/blob/b2974a26cf42c262d84662e32869b0e2c3950216/pymwp/monomial.py\#L95}

\newcommand{\hrefdeclaration}{https://statycc.github.io/pymwp/demo/\#basics_inline_variable.c}
 \allowdisplaybreaks[1]
\definecolor{clem}{HTML}{60656F} \definecolor{thos}{HTML}{BABD8D} \definecolor{thor}{HTML}{EB6424} \definecolor{neea}{HTML}{823038}

\title{mwp-Analysis Improvement and Implementation: Realizing Implicit Computational Complexity}
\titlerunning{mwp-Analysis Improvement and Implementation: Realizing Implicit Complexity} %

\author{Clément Aubert}{School of Computer and Cyber Sciences, Augusta University \and \url{https://spots.augusta.edu/caubert/}}{caubert@augusta.edu}{https://orcid.org/0000-0001-6346-3043}{}
\author{Thomas Rubiano}{LIPN – UMR 7030 Université Sorbonne Paris Nord \and \url{https://people.irisa.fr/Thomas.Rubiano/}}{rubiano.thomas@gmail.com}{}{}
\author{Neea Rusch}{School of Computer and Cyber Sciences, Augusta University \and \url{https://nkrusch.github.io/}}{nrusch@augusta.edu}{https://orcid.org/0000-0002-7354-5330}{}
\author{Thomas Seiller}{LIPN – UMR 7030 Université Sorbonne Paris Nord \and CNRS \and \url{https://www.seiller.org/}}{seiller@lipn.fr}{https://orcid.org/0000-0001-6313-0898}{}

\funding{This research is supported by the \href{https://face-foundation.org/higher-education/thomas-jefferson-fund/}{Th. Jefferson Fund} of the Embassy of France in the United States and the \href{https://face-foundation.org/}{FACE Foundation}, and has benefited from the research meeting \href{https://www.dagstuhl.de/de/programm/kalender/evhp/?semnr=21453}{21453 \enquote{Static Analyses of Program Flows: Types and Certificate for Complexity}} in Schloss Dagstuhl. Th.\ Rubiano and Th.\ Seiller are supported by the Île-de-France region through the DIM RFSI project \enquote{CoHOp}.}

\acknowledgements{The authors wish to express their gratitude to Assya Sellak for her contribution to this work, and to previous reviewers for their comments.}

\authorrunning{C. Aubert, Th. Rubiano, N. Rusch and Th. Seiller}

\Copyright{Clément Aubert, Thomas Rubiano, Neea Rusch and Thomas Seiller}



\category{}
\relatedversion{}
\supplement{}
\supplementdetails[subcategory ={Source Code}]{Software}{https://github.com/statycc/pymwp}
\supplementdetails[subcategory={Documentation and Demo}]{Software}{https://statycc.github.io/pymwp}

\EventEditors{John Q. Open and Joan R. Access}
\EventNoEds{2}
\EventLongTitle{42nd Conference on Very Important Topics (CVIT 2016)}
\EventShortTitle{CVIT 2016}
\EventAcronym{CVIT}
\EventYear{2016}
\EventDate{December 24---27, 2016}
\EventLocation{Little Whinging, United Kingdom}
\EventLogo{}
\SeriesVolume{42}
\ArticleNo{23}

\makeatletter%
\begin{document}

\maketitle
\begin{abstract}
	Implicit Computational Complexity (ICC) drives better understanding of complexity classes, but it also guides the development of resources-aware languages and static source code analyzers.
	Among the methods developed, the \emph{mwp-flow analysis}~\cite{Jones2009} certifies polynomial bounds on the size of the values manipulated by an imperative program.
	This result is obtained by bounding the transitions between states instead of focusing on states in isolation, as most static analyzers do, and is not concerned with termination or tight bounds on values.
	Those differences, along with its built-in compositionality, make the mwp-flow analysis a good target for determining how ICC-inspired techniques diverge compared with more traditional static analysis methods.
	This paper's contributions are three-fold: we fine-tune the internal machinery of the original analysis to make it tractable in practice; we extend the analysis to function calls and leverage its machinery to compute the result of the analysis efficiently; and we implement the resulting analysis as a lightweight tool to automatically perform data-size analysis of \texttt{C} programs.
	This documented effort prepares and enables the development of certified complexity analysis, by transforming a costly analysis into a tractable program, that furthermore decorrelates the problem of deciding if a bound exist with the problem of computing it.
\end{abstract}

\section{Introduction: letting ICC drive the development of static analyzers}
\label{sec:intro}

Certifying program resource usages is possibly as crucial as the specification of program correctness, since a guaranteed correct program whose memory usage exceeds available resources is, in fact, unreliable.
The field of Implicit Computational Complexity (ICC) theory~\cite{DalLago2012a} pioneers in \enquote{embedding} in the program itself a guarantee of its resource usage, using \eg bounded recursion~\cite{Bellantoni1992,Leivant1993}
or type systems~\cite{Baillot2004,Lafont2004}.
This field initiated numerous distinct and original approaches, primarily to characterize complexity classes in a machine-independent way, with increasing expressivity, but these approaches have rarely materialized into concrete programming languages or program analyzers: even if, as opposed to traditional complexity, its models are generally expressive enough to write down actual algorithms~\cite[p.~11]{Moyen2017c}, they rarely escape the sphere of academia or extend beyond toy languages, with a few exceptions~\cite{Avanzini2017,Hoffmann2012b}.
However, by abstracting away constant factors and insignificant orders of magnitude, it is frequently conjectured that ICC will allow sidestepping some of the difficult issues one usually has to face when inferring the resource usage of a concrete program.

This work reinforces this conjecture by adjusting, improving and implementing an existing ICC technique, the \emph{mwp-bounds analysis}~\cite{Jones2009}, that certifies that the values computed by an imperative program will be bounded by polynomials in the program's input.
This flow analysis is elegant but computationally costly, and it missed an opportunity to leverage its built-in compositionality: we address both issues by revisiting and expanding the original flow calculus, and further make our point by implementing it on a subset of the \texttt{C} programming language.
While the theory has been improved to allow analysis of function definitions and calls---including recursive ones, a feature not widely supported~\cite[p.~359]{Hainry2021}---, its integration into the implementation is underway, as we placed primary focus on developing an efficient and implementable technique for program analysis.
Implementing a tool along the theory enabled testing improvements in real-life, which in return drove adjustments to the theory.

Our enhanced technique answers positively two questions asked by the authors of the original analysis~\cite[Section 1.2]{Jones2009}, namely
\begin{enumerate*}
	\item Can the method be extended to richer languages?
	\item Can it lead to powerful and convenient tools?
\end{enumerate*}
It also supports the conjecture that ICC can be used to construct concrete tools, but highlights that doing so requires adjusting the theory to make it tractable in practice.
This work also provides better insight into the original analysis, by \eg separating the algorithm to decide the existence of a bound from its evaluation into a concrete bound; and by illustrating its plasticity: while our analysis purely extends the original one, it nevertheless greatly alters its internal machinery to ease its implementability.
Last but not least, our technique is orthogonal to most static analysis methods, that focus on worst-case resource-usage complexity or termination, while ours establishes that the growth rate of variables is at most polynomially related to their inputs.

\makeatletter
\def\@pdfborder{0 0 0}%
\def\@pdfborderstyle{/S/U/W 1}%
\makeatother

\section{Background: the original flow analysis}
The original analysis~\cite{Jones2009} computes a polynomial bound---if it exists---on the sizes of variables in an imperative \texttt{while} programming language, extended with a \texttt{loop} operator, by computing for each variable a vector that tracks how it depends from other variables---and the program itself gets assigned a matrix collecting those vectors.
While this does not ensure termination, it provides a certificate guaranteeing that the program uses at most a polynomial amount of space, and as a consequence that \emph{if} it terminates, it will do so in polynomial time.

\subsection{Language analyzed: fragments of imperative language}
\label{subsec:language}
\begin{definition}[Imperative Language]
	\label{def:lang}
	Letting variables range over \pr|X| and \pr|Y| and boolean expressions over \pr{b}, we define \emph{expressions} \pr|e| and \emph{commands} \pr|C| as follows:
	\begin{align*}
		\text{\pr|e|} \coloneqq & \text{\pr|X|} \BNF \text{\pr|X - Y|} \BNF \text{\pr|X + Y|} \BNF \text{\pr|X * Y|}                                                           \\
		\text{\pr|C|} \coloneqq & \text{\pr|X = e|} \BNF \text{\pr|if b then C else C|} \BNF \text{\pr|while b do \{C\}|} \BNF \text{\pr|loop X \{C\}|} \BNF \text{\pr|C ; C|}\end{align*}

	where \pr|loop X {C}| means \enquote{do \pr|{C}| \pr|X| times} and \pr|C;C| is used for sequentiality	(\enquote{do \pr|C|, then \pr|C|}).
	We write \enquote{program} for a series of commands composed sequentially.
\end{definition}

This language assumes that the program's inputs are the only variables, and that assigning a value to a variable inside the program is not permitted.
Extending flow calculi to those operations has been discussed~\cite[p.~3]{Jones2009} and proven possible~\cite{BenAmram2010}, but we leave this for future work---in particular, our \texttt{C} examples will be of \pr|foo| functions with their variables listed as parameters\footnote{Our implementation allows to relax this condition, as exemplified in \href{\hrefdeclaration}{\texttt{inline\_variable.c}}, without losing any of the results expressed in this paper. Assuming a fixed number of variables, known ahead of time, is mostly a theoretical artifact used to simplify the analysis. \label{int-var}}.
However, we disallow \wlg composed expressions of the form \pr{X + Y * Y}, that can always be dealt with in the style of three-address code.

\subsection{A flow calculus of mwp-bounds for complexity analysis}
\label{ssec:a-flow-calculus}

\emph{Flows} characterize controls from one variable to another, and can be, in increasing growth rate, of type \(0\)---the absence of any dependency---\emph{m}aximum, \emph{w}eak polynomial and \emph{p}olynomial.
The bounds on programs written in the syntax of \autoref{subsec:language} are represented and calculated thanks to vectors and matrices whose coefficients are elements of the mwp semi-ring.

\begin{definition}[The mwp semi-ring and matrices over it]
	\label{def:mwp-matrix-alg}
	Letting \(\mwp = \{0, m, w, p\}\) with \(0 < m < w < p\), and \(\alpha\), \(\beta\), \(\gamma\) range over \mwp, the \emph{mwp semi-ring} \((\mwp, 0, m, +, \times)\) is defined with \(+ = \max\), \(\alpha \times \beta = \max (\alpha, \beta)\) if \(\alpha, \beta \neq 0\), and \(0\) otherwise.

	We denote \(\mathbb{M}(\mwp)\) the matrices over \mwp, and, fixing \(n\in\mathbb{N}\), \(M\) for \(n \times n\) matrices over \mwp, \(M_{ij}\) for the coefficient in the \(i\)th row and \(j\)th column of \(M\), \(\oplus\) for the componentwise addition, and \(\otimes\) for the product of matrices defined in a standard way.
	The \(\zmat\)-element for addition is \(\zmat_{ij} = 0\) for all \(i, j\), and the \(\umat\)-element for product is \(\umat_{ii} = m\), \(\umat_{ij} = 0\) if \(i \neq j\), and the resulting structure
	\((\mathbb{M}(\mwp), \zmat, \umat, \otimes, \oplus)\) is a semi-ring that we simply write \(\mathbb{M}(\mwp)\).
	The closure operator \(\cdot^{*}\) is \(M^* = \umat \oplus M \oplus (M^2) \oplus \hdots\), for \(M^0 = \umat\), \(M^{m+1} = M \otimes M^m\).
\end{definition}

Although not crucial to understand our development, details about (strong) semi-rings and the mwp semi-ring are in \autoref{sec:app:mwp}, and the construction of a semi-ring whose elements are matrices with coefficients in a semi-ring---so, in particular, \(\mathbb{M}(\mwp)\)---is given in \autoref{sec:app:matrix}.

Below, we let \(V_1\), \(V_2\) be column vectors with values in \mwp, \(\alpha V_1\) be the usual scalar product, and \(V_1 \oplus V_2\) be defined componentwise.
We write \(\{_{i}^{\alpha}\}\) for the vector with \(0\) everywhere except for \(\alpha\) in its \(i\)th row, and \(\{_{i}^{\alpha} , _{j}^{\beta}\}\) for \(\{_{i}^{\alpha}\} \oplus \{_{j}^{\beta}\}\).

Replacing in a matrix \(M\) the \(j\)th column vector by \(V\) is denoted \(M \xleftarrow{j} V\).
The matrix \(M\) with \(M_{ij} = \alpha\) and \(0\) everywhere else is written \(\{_{i}^{\alpha} \rightarrow j\}\), and the set of variables in the expression \pr|e| is written \(\var(\text{\pr|e|})\).
The assumption is made that exactly \(n\) different variables are manipulated throughout the analyzed program, so that \(n\)-vectors are assigned to expressions---in a non-deterministic way, to capture larger classes of programs~\cite[Section 8]{Jones2009}---and \(n \times n\) matrices are assigned to commands using the rules presented \autoref{fig:orig-rules}~\cite[Section 5]{Jones2009}.

\begin{figure}
	\begin{subfigure}{\textwidth}
		\begin{centering}
			\begin{tabular}{c c c}
				\begin{prooftree}[small]
					\infer0[E1]{ \vdashJK \text{\pr|Xi|} : \{_{\text{\pr|i|}}^{m}\}}
				\end{prooftree}
				 & \hspace{1em} &
				\begin{prooftree}[small]
					\infer0[E2]{ \vdashJK \text{\pr{e}} : \{ _{\text{\pr|i|}}^{w} \mid \text{\pr|Xi|} \in \var(\text{\pr{e}}) \}}
				\end{prooftree}
				\\[1.2em]
				\begin{prooftree}[small]
					\hypo{\vdashJK \text{\pr{Xi}} : V_1}
					\hypo{\vdashJK \text{\pr{Xj}} : V_2}
					\infer[left label={\(\star\in\{+, -\}\)}]2[E3]{\vdashJK \text{\pr|Xi $\star$ Xj|} : pV_1 \oplus V_2}
				\end{prooftree}
				 &              &
				\begin{prooftree}[small]
					\hypo{\vdashJK \text{\pr{Xi}} : V_1}
					\hypo{\vdashJK \text{\pr{Xj}} : V_2}
					\infer[left label={\(\star\in\{+, -\}\)}]2[E4]{\vdashJK \text{\pr|Xi $\star$ Xj|} : V_1 \oplus pV_2}
				\end{prooftree}
			\end{tabular}
			\caption{Rules for assigning vectors to expressions}
			\label{fig:rules-expressions}
		\end{centering}
	\end{subfigure}
	\\[1.2em]
	\begin{subfigure}{\textwidth}
		\begin{centering}
			\begin{prooftree}[small]
				\hypo{ \vdashJK \text{\pr{e}} : V}
				\infer1[A]{\vdashJK \text{\pr|Xj = e|} : \umat \xleftarrow{\text{\pr|j|}} V}
			\end{prooftree}
			\hspace{1.6em}
			\begin{prooftree}[small]
				\hypo{ \vdashJK \text{\pr{C1}} : M_1}
				\hypo{ \vdashJK \text{\pr{C2}} : M_2}
				\infer2[C]{\vdashJK \text{\pr{C1; C2}} : M_1 \otimes M_2}
			\end{prooftree}
			\\[1.2em]
			\begin{prooftree}[small]
				\hypo{ \vdashJK \text{\pr{C1}} : M_1}
				\hypo{ \vdashJK \text{\pr{C2}} : M_2}
				\infer2[I]{\vdashJK \text{\pr|if b then C1 $\text{\hspace{.3em}}$ else C2|} : M_1 \oplus M_2} \end{prooftree}
			\\[1.2em]
			\begin{prooftree}[small]
				\hypo{\vdashJK \text{\pr|C|} : M}
				\infer[left label={\(\forall i, M_{ii}^* = m\)}]1[L]{\vdashJK \text{\pr|loop Xl \{C\}|} : M^* \oplus \{_{\text{\pr|l|}}^{p}\rightarrow j \mid \exists i, M_{ij}^* = p\}}
			\end{prooftree}
			\\[1.2em]
			\begin{prooftree}[small]
				\hypo{\vdashJK \text{\pr|C|} : M}
				\infer[left label={\(\forall i, M_{ii}^* = m\) and \(\forall i, j, M^*_{ij} \neq p\)}]1[W]{\vdashJK \text{\pr|while b do \{C\}|} : M^*}
			\end{prooftree}
			\caption{Rules for assigning matrices to commands}
			\label{fig:rules-commands}
		\end{centering}
	\end{subfigure}
	\caption{Original non-deterministic (\enquote{\textsc{J}ones-\textsc{K}ristiansen}) flow analysis rules}
	\label{fig:orig-rules}
\end{figure}

The intuition is that if \(\vdashJK \text{\pr|C|}:M\) can be derived, then all the values computed by \pr|C| will grow at most polynomially \wrt its inputs~\cite[Theorem 5.3]{Jones2009}, \eg will be bounded by \(\max(\vec{x}, p_1(\vec{y})) + p_2(\vec{z})\), where \(p_1\) and \(p_2\) are polynomials and \(\vec{x}\) (\resp \(\vec{y}\), \(\vec{z}\)) are \(m\)-(\resp \(w\)-, \(p\)-)annotated variables in the vector for the considered output.
Since the derivation system is non-deterministic, multiple matrices and polynomial bounds---that sometimes coincide---may be assigned to the same program.
Furthermore, the coefficient at \(M_{\text{\pr|ij|}}\) carries quantitative information about the way \pr|Xi| depends on \pr|Xj|, knowing that \(0\)- and \(m\)-flows are harmless and without constraints, but that \(w\)- and \(p\)- flows are more harmful \wrt polynomial bounds and need to be handled with care, particularly in loops---hence the condition on the L and W rules.
The derivation may fail---some programs may not be assigned a matrix---, if at least one of the variables used in the body of a loop depends \enquote{too strongly} upon another, making it impossible to ensure polynomial bounds on the loop itself.
We will use the following example as a common basis to discuss possible failure, non-determinism, and our
improvements.

\begin{example}

	\label{ex:start}
	Consider \pr|loop X3{X2 = X1 + X2}|.
	The body of the \pr|loop| command admits three different derivations, obtained by applying A to one of the three derivation of the expression \pr|X1 + X2|, that we name \(\pi_0\), \(\pi_1\) and \(\pi_2\):

	{\vspace{.5em} %
	\centering
	\scalebox{0.85}{
		\begin{prooftree}[small]
			\infer0[E1]{ \vdashJK \text{\pr|X1|} : \mat{m\\0\\0}}
			\infer0[E1]{ \vdashJK \text{\pr|X2|} : \mat{0\\m\\0}}
			\infer2[E3]{ \vdashJK \text{\pr|X1 + X2|} : \mat{p\\m\\0}}
		\end{prooftree}
	}
	\hfill
	\scalebox{0.85}{
		\begin{prooftree}[small]
			\infer0[E1]{ \vdashJK \text{\pr|X1|} : \mat{m\\0\\0}}
			\infer0[E1]{ \vdashJK \text{\pr|X2|} : \mat{0\\m\\0}}
			\infer2[E4]{ \vdashJK \text{\pr|X1 + X2|} : \mat{m\\p\\0}}
		\end{prooftree}
	}
	\hfill
	\scalebox{0.85}{
		\begin{prooftree}[small]
			\infer0[E2]{ \vdashJK \text{\pr|X1 + X2|} : \mat{w\\w\\0}}
		\end{prooftree}
	}
	}

	From \(\pi_0\), the derivation of \pr|loop X3{X2 = X1 + X2}| can be
	completed using A and L, but since L requires having only \(m\) coefficients
	on the diagonal, \(\pi_1\) cannot be used to complete the derivation,
	because of the \(p\) coefficient in a box below:

	\scalebox{0.85}{
		\begin{prooftree}[small]
			\hypo{}
			\ellipsis{\(\pi_0\)}{ \vdashJK \text{\pr|X1 + X2|} : \mat{p\\m\\0}}
			\infer1[A]{ \vdashJK \text{\pr|X2 = X1 + X2|} : \mat{m&p&0 \\ 0&m&0 \\0&0&m}}
			\infer1[L]{ \vdashJK \text{\pr|loop X3 \{X2 = X1 + X2\}|} : \mat{m&p&0 \\ 0&m&0 \\0&p&m}}
		\end{prooftree}
	}
	\hfill \begin{prooftree}[small]
		\hypo{}
		\ellipsis{\(\pi_1\)}{ \vdashJK \text{\pr|X1 + X2|} : \mat{m\\p\\0}}
		\infer1[A]{ \vdashJK \text{\pr|X2 = X1 + X2|}: \mat{m&m&0 \\ 0& \boxed{p} &0 \\0&0&m}}
	\end{prooftree}

	Similarly, using A after \(\pi_2\) gives a \(w\) coefficient on the diagonal and makes it impossible to use L, hence only one derivation for this program exists.
\end{example}

\subsection{Limitations and inefficiencies of the mwp analysis}

Even if the proof techniques are far from trivial, with only 9 rules and skipping over boolean expressions (observe that the condition \pr|b| has no impact in the rules I or W), the analysis is flexible and easy to carry---at least mathematically.
It also has inherent limitations: while the technique is sound, it is not complete and programs such as \href{https://statycc.github.io/pymwp/demo/#other_gcd.c}{greatest common divisor} fail to be assigned a matrix.
We will discuss in \autoref{ssec:lit-rev}, the benefits and originality of this analysis, but we would now like to stress how it is computationally inefficient, since the non-determinacy makes the analysis costly to carry and can lead to memory explosions.

Abstracting \autoref{ex:start}, one can see that the base case of non-determinism---\eg, to assign a vector to \pr|X1 $\star$ X2|--yields vectors %
\(\mat{p \\ m}\) (using E1 then E3), \(\mat{m \\ p}\) (using E1 then E4) and \(\mat{w \\ w}\) (using E2).
Since none of those vectors is less than the others, only two strategies are available to analyze a larger program containing \pr|X1 $\star$ X2|: either the derivations for this base case are considered one after the other, or they are all stored in memory at the same time.
Considering the derivations for the base case one after the other can lead to a time explosion, as a program of \(n\) lines can have \(3^n\) different derivations---as exemplified by \href{\hrefexplosion}{\texttt{explosion.c}}, a simple series of applications---and it is possible that only one of them can be completed, so all must be explored.
On the other hand, storing those three vectors and constructing all the matrices in parallel leads to a memory explosion: the analysis for two commands involving 6 variables, with 3 choices---that cannot be simplified as explained previously---would result in 9 matrices of size \(6 \times 6\), \ie, 324 coefficients.
All in all, a program of \(n\) lines with \(x\) different variables can require \(c_1^n\) different derivations, that can produce up to \((c_2 \times x)^2\) coefficients to store for some constants \(c_1\), \(c_2\).

Beyond inefficiency, there are additional limitations: while the analysis is naturally compositional, this feature is not leveraged in the original system; furthermore, an occurrence of non-polynomial flows in the matrix causes the analysis to simply stop, thus not capturing failure in a meaningful way. We will discuss our solutions to these deficiencies next.

\section{A deterministic, always-terminating, declension of the mwp analysis}

The problem of finding a derivation in the original calculus is in NP~\cite[Theorem 8.1]{Jones2009}.
But since all the non-determinism is in the rules to assigning a vector, the potentially exponential number of derivations are actually extremely similar.
Hence, instead of having the analysis stop when failing to establish a derivation and re-starting from scratch, storing the different vectors and constructing the derivation while keeping all the options open seems to be a better strategy, but, as we have seen, this causes a memory blow-up.
We address it by fine-tuning the internal machinery: to represent non-determinism, we let the matrices take as values either functions from choices to coefficients in \mwp or coefficients in \mwp, so that instead of mapping choices to derivations, all the derivations are represented by the same matrix that internalizes the different choices.
\autoref{ssec:choice-semi-ring} discusses this improvement, that results in a notable gain: the (unique) matrix produced for a program involving 6 variables, with 3 choices, is a \(6\times 6\) matrix that can be represented by 66 coefficients instead of the 324 we previously had---this is because 30 coefficients are \enquote{simple} values in \mwp, and 6 are functions from a set of choices \(\{0, 1, 2\}\) to values in \mwp, each represented with 6 coefficients.

For the choices that give coefficients fulfilling the side condition of L or W, the derivation can proceed as usual, but when a particular choice gives a coefficient that violates it, we decided against simply removing it.
Instead, to guarantee that all derivations always terminate, we mark that choice by indicating that it would not provide a polynomial bound.
This requires extending the \mwp semi-ring with a special value \(\infty\) that represents failure in a local way, marking non-polynomial flows, and is detailed in \autoref{ssec:infinity-semi-ring}.
As a by-product, this allows gaining fine-grained information on programs that \emph{do not} have polynomially bounded growth, since the precise dependencies that break this growth rate can be localized.

Taken together (\autoref{ssec:merge}), our improvements ensure that exactly one matrix will always be assigned to a program while carrying over the correctness of the original analysis.
We give in \autoref{fig:new-rules} the deterministic system we are introducing in full, but will gently introduce it though the remaining parts of this section: note that the rules A, C and I are unchanged, up to the fact that the matrices, sum and product are in a different semi-ring.

\begin{figure}
	\begin{subfigure}{\textwidth}
		\begin{centering}

			\begin{prooftree}[small]
				\infer[left label={\(\star\in\{+, -\}\)}]0[E\(^{\textsc{A}}\)]{\vdash \text{\pr|Xi $\star$ Xj|} : (0 \mapsto \{_{\text{\pr|i|}}^{m}, _{\text{\pr|j|}}^{p}\}) \oplus (1 \mapsto \{_{\text{\pr|i|}}^{p}, _{\text{\pr|j|}}^{m}\}) \oplus (2 \mapsto \{_{\text{\pr|i|}}^{w}, _{\text{\pr|j|}}^{w}\})}
			\end{prooftree}
			\\[1.2em]
			\begin{prooftree}[small]
				\infer0[E\(^{\textsc{M}}\)]{ \vdash \text{\pr|Xi * Xj|} : \{ _{\text{\pr|i|}}^{w}, _{\text{\pr|j|}}^{w} \}}
			\end{prooftree}
			\hspace{4em}
			\begin{prooftree}[small]
				\infer0[E\(^{\textsc{S}}\)]{ \vdash \text{\pr|Xi|} : \{ _{\text{\pr|i|}}^{m}\}}
			\end{prooftree}

			\caption{Rules for assigning vectors to expressions}
			\label{fig:new-rules-expressions}
		\end{centering}
	\end{subfigure}
	\begin{subfigure}{\textwidth}
		\begin{centering}
			\begin{prooftree}[small]
				\hypo{ \vdash \text{\pr{e}} : V}
				\infer1[A]{\vdash \text{\pr|Xj = e|} : \umat \xleftarrow{\text{\pr|j|}} V}
			\end{prooftree}
			\hfill
			\begin{prooftree}[small]
				\hypo{ \vdash \text{\pr{C1}} : M_1}
				\hypo{ \vdash \text{\pr{C2}} : M_2}
				\infer2[C]{\vdash \text{\pr{C1; C2}} : M_1 \otimes M_2}
			\end{prooftree}
			\hfill
			\begin{prooftree}[small]
				\hypo{ \vdash \text{\pr{C1}} : M_1}
				\hypo{ \vdash \text{\pr{C2}} : M_2}
				\infer2[I]{\vdash \text{\pr|if b then C1 $\text{\hspace{.3em}}$ else C2|} : M_1 \oplus M_2} \end{prooftree}
			\\[1.2em]
			\begin{prooftree}[small]
				\hypo{ \vdash \text{\pr|C|}: M }
				\infer1[L\(^{\infty}\)]{\vdash \text{\pr|loop Xl \{C\}|} : M^* \oplus \{_{j}^{\infty} \rightarrow j \mid M^*_{jj} \neq m\} \oplus \{_{\text{\pr|l|}}^{p}\rightarrow j \mid \exists i, M_{ij}^* = p\} } \end{prooftree}
			\\[1.2em]
			\begin{prooftree}[small]
				\hypo{ \vdash \text{\pr|C|}: M }
				\infer1[W\(^{\infty}\)]{\vdash \text{\pr|while b do \{C\}|} : M^* \oplus \{_{j}^{\infty} \rightarrow j \mid M^*_{jj} \neq m\} \oplus \{_{i}^{\infty}\rightarrow j \mid M_{ij}^* = p\} } \end{prooftree}
		\end{centering}
		\caption{Rules for assigning matrices to commands}
		\label{fig:new-rules-commands}
	\end{subfigure}
	\caption{Deterministic improved flow analysis rules}
	\label{fig:new-rules}
\end{figure}

\subsection{Internalizing non-determinism: the choice data flow semi-rings}
\label{ssec:choice-semi-ring}

Internalizing the choice requires altering the semi-ring used in the analysis: we want to replace the three vectors over \(\mwp\) that can be assigned to an expression by a single vector over \(\{0, 1, 2\} \to \mwp\) that captures the same three choices.
For a program needing to decide \(p\) times between the 3 available choices, this means replacing the \(3 \times p\) different matrices in \(\mathbb{M}(\mwp)\) by a single matrix in \(\mathbb{M} (\{0, 1, 2\}^p \to \mwp)\).
As proven in \autoref{sec:app:choice}, for any strong semi-ring \(\mathbb{S}\) and family of sets \((A_i)_{i=1,\dots,p}\), both \(A_i \to \mathbb{S}\) and \(\mathbb{M} (\prod_{i=1}^p A_i \to \mathbb{S})\) are semi-rings, using the usual cartesian product of sets, and there exists an isomorphism
\(
\mathbb{M} (\prod_{i=1}^p A_i \to \mathbb{S} ) \cong \prod_{i=1}^p A_i \to \mathbb{M} (\mathbb{S} )
\).
This dual nature of the semi-ring considered is useful:
\begin{itemize}
	\item the analysis will now assign an element \(M\) of \(\mathbb{M} (\prod_{i=1}^p A_i \to \mwp )\) to a program;
	\item representing \(M\) as an element of \(\prod_{i=1}^p A_i \to \mathbb{M} (\mwp )\) allows to use an \emph{assignment} \(\vec{a}=(a_1,\dots,a_p)\in \prod_{i=1}^p A_i\) to produce a matrix \(M[\vec{a}]\in \mathbb{M} (\mwp )\), recovering the mwp-flow that would have been computed by making the choices \(a_1,\dots, a_p\) in the derivation.
\end{itemize}

\begin{remark}
	As the unique degree of non-determinism to assign a matrix to commands is 3, our modification of the analysis flow consists simply of recording the different choices by letting \(A_i = \{0, 1, 2\}\) for all \(i=1,\dots,p\) where \(p\)	is the number of times a choice had to be taken.
	Starting with \autoref{composition}, function calls will require potentially different sets \(A_i\).
\end{remark}

\begin{notation}
	\label{notation-delta}
	In the following and in the implementation alike, we will denote a function \((a_1^0 \times \cdots \times a_p^0 \mapsto \alpha_0) + \cdots + (a_1^k \times \cdots \times a_p^k \mapsto \alpha_k)\) in \(A^p \to \mwp\) with \(\Card(A) = k\) by, omitting the product, \((\alpha_0\delta(a_1^0, 0) \cdots \delta(a_p^0, p) ) + \cdots + (\alpha_k\delta(a_1^k, n)\cdots \delta(a_p^k, p))\), with \(\delta(i, j) = m\) if the \(j\)th choice is \(i\), \(0\) otherwise.
	\autoref{ex:mixing} will justify and explain this choice.
\end{notation}

Our derivation system replaces the E3 and E4 rules with a single rule E\(^{\textsc{A}}\) (\enquote{additive}), and splits E2 in two exclusive rules, E\(^{\textsc{M}}\) for \enquote{multiplicative} and
E\(^{\textsc{S}}\) for \enquote{simple} (atomic) expressions---\autoref{lem:equi-expr} will prove how they are equivalent. %

\begin{example}
	We represent the vectors \(\mat{p\\m\\0}\),	\(\mat{m\\p\\0}\) and \(\mat{w\\w\\0}\) from \autoref{ex:start} with a single vector\(\mat{p \delta(0, 0) + m\delta(1, 0) + w \delta(2, 0) \\ m \delta(0, 0) + p \delta(1, 0) + w \delta(2, 0) \\ 0}\), that can be read as \(\mat{\{ 0 \mapsto p, 1 \mapsto m, 2 \mapsto w\} \\ \{0 \mapsto m, 1 \mapsto p, 2 \mapsto w\} \\ 0}\), where we write \(0\) for \(\{0 \mapsto 0, 1 \mapsto 0, 2 \mapsto 0\}\)\footnote{The implementation supports both coefficients from \mwp \emph{and} coefficients from \(\{0, 1, 2\}^p \to \mwp\), \cf \eg a simple assignment example \href{\hrefexassign}{\texttt{assign\_expression.c}}.}.
	Since in particular\footnote{This is a variant of \autoref{lem:semi-ring-iso} from \autoref{sec:app:choice}. While the latter lemma applies to algebras of square matrices, a similar result holds for rectangular matrices of a fixed size; the algebraic structure is no longer that of a semi-ring as rectangular matrices do not possess a proper multiplication, but the proof can be adapted to show the existence of an isomorphism of modules between the considered spaces.},
	\(\mathbb{M}(\{0, 1, 2\} \to\mwp) \cong \{0, 1, 2\} \to \mathbb{M}(\mwp)\), the obtained vector can be rewritten as	\(0 \mapsto \mat{p\\m\\0}, 1 \mapsto \mat{m\\p\\0}, 2 \mapsto \mat{w\\w\\0}\).
\end{example}

\subsection{Internalizing failure: de-correlating derivations and bounds}
\label{ssec:infinity-semi-ring}

The original analysis stops when detecting a non-polynomial flow, puts an end to the chosen strategy (\ie set of choices) and restarts from scratch with another one.
We adapt the rules so that every derivation can be completed even in the presence of non-polynomial flows, thanks to a new top element, \(\infty\), representing failure in a local way.

Ignoring our previous modification in this subsection, the semi-ring \(\mwp^\infty\) we need to consider is \((\mwp \cup \{\infty\}, 0, m, +^{\infty}, \times^{\infty})\), with \(\infty > \alpha\)
for all \(\alpha \in \mwp\), \(+^{\infty}=\max\) as before, and \(\alpha \times^{\infty} \beta = 0\) if \(\alpha, \beta \neq \infty\) and \(\alpha\) or \(\beta\) is \(0\), \(\max(\alpha, \beta)\) otherwise.
This different condition in the definition of \(\times^{\infty}\) ensures that once non-polynomial flows have been detected, they cannot be erased (as \(\infty \times^{\infty} 0 = \infty\), some additional details are discussed in \autoref{sec:app:partiality}).

The only cases where the original analysis may fail is if the side conditions of L or W (\autoref{fig:orig-rules}) are not met.
We replace those by L\(^\infty\) and W\(^\infty\) (\autoref{fig:new-rules}), which replace the problematic coefficients with \(\infty\), marking non-polynomial dependencies, and carry on the analysis.

\begin{example}
	The program from \autoref{ex:start} would now receive three derivations (omitting the one obtained from \(\pi_0\), as the resulting matrix is identical):

	{\small
	\begin{prooftree}[small]
		\hypo{}
		\ellipsis{\(\pi_1\)}{ \vdash \text{\pr|X1 + X2|} : \mat{m\\p\\0}}
		\infer1[A]{ \vdash \text{\pr|X2 = X1 + X2|}: \mat{m&m&0 \\ 0& p &0 \\0&0&m}}
		\infer1[L\(^\infty\)]{ \vdash \text{\pr|loop X3 \{X2 = X1 + X2\}|} : \mat{m&p&0 \\ 0&\infty & 0 \\ 0 & p & m}}
	\end{prooftree}
	\hfill \begin{prooftree}[small]
		\infer0[E2]{\vdash \text{\pr|X1 + X2|} : \mat{w \\ w \\0} }
		\infer1[A]{ \vdash \text{\pr|X2 = X1 + X2|}: \mat{m& w &0 \\ 0 & w &0 \\0 & 0 & m}}
		\infer1[L\(^\infty\)]{ \vdash \text{\pr|loop X3 \{X2 = X1 + X2\}|} : \mat{m & w &0 \\ 0 & \infty & 0 \\0 & 0& m }}
	\end{prooftree}
	}

	Of course, neither of those two derivations would yield polynomial bound---since they contain \(\infty\) coefficients---but it becomes possible to determine that the last one is \enquote{better}---since \(\mat{p \\ \infty \\ p} > \mat{w \\ \infty \\ 0}\)---and to observe how their \enquote{failure} would propagate in larger programs, possibly establishing that one fares better than the other in terms of non-polynomial growths.
\end{example}

\subsection{Merging the improvements: illustrations and proofs}
\label{ssec:merge}

We prove that our system captures the original system in the sense that set aside \(\infty\) coefficients, both systems agree (\autoref{lem:equi-expr}), but also that exactly one matrix is produced per program (\autoref{lem:det-term})---\ie that we can analyze as many programs as originally, and still be correct regarding the bounds.
Before doing so, we would like to give more specifics on our system, by combining the semi-rings and intuitions from the previous two subsections.
We have discussed our \enquote{axiomatic} (E\(^{\textsc{A}}\), E\(^{\textsc{M}}\), E\(^{\textsc{S}}\)) and \enquote{loop} rules (L\(^\infty\) and W\(^\infty\)), but remain to discuss the rules for assignment (A), \pr|if| (I) and composition (C)---which is where both improvements meet.
Mathematically speaking, adopting the semi-ring defined over matrices with coefficients in \(\{0, 1, 2\}^p \to \mwp \cup \{\infty\}\) is straightforward, and we simply write \(\oplus\) and \(\otimes\) the operations resulting from merging the two transformations.
We discuss in \autoref{sec:advantage} how, however, those operations are computationally costly and how we address this challenge.

\begin{example}
	\label{ex:mixing}
	Using our deterministic system presented in \autoref{fig:new-rules}, consider the following:
	\begin{center}\small
		\begin{prooftree}[small]
			\infer0[E\(^{\textsc{A}}\)]{\vdash \text{\pr|X1 + X2|} : V}
			\infer1[A]{\vdash \text{\pr|X1 = X1 + X2|} : 1 \xleftarrow{1} V}
			\infer0[E\(^{\textsc{A}}\)]{\vdash \text{\pr|X1 - X3|} : V'}
			\infer1[A]{\vdash \text{\pr|X1 = X1 - X3|} : 1 \xleftarrow{1} V'}
			\infer2[I]{\vdash \text{\pr|if b then \{X1 = X1 + X2\} else \{X1 = X1 - X3\}|} : (1 \xleftarrow{1} V) \oplus (1 \xleftarrow{1} V')}
		\end{prooftree}
	\end{center}
	with
	\(\begin{aligned}[t]
		V                                                                  & = 0 \mapsto \{_{\text{\pr|1|}}^{m}, _{\text{\pr|2|}}^{p}\} \oplus 1 \mapsto \{_{\text{\pr|1|}}^{p}, _{\text{\pr|2|}}^{m}\} \oplus 2 \mapsto \{_{\text{\pr|1|}}^{w}, _{\text{\pr|2|}}^{w}\}         \\
		V'                                                                 & = 0 \mapsto \{_{\text{\pr|1|}}^{m}, _{\text{\pr|3|}}^{p}\} \oplus 1 \mapsto \{_{\text{\pr|1|}}^{p}, _{\text{\pr|3|}}^{m}\} \oplus 2 \mapsto \{_{\text{\pr|1|}}^{w}, _{\text{\pr|3|}}^{w}\}         \\
		1 \xleftarrow{1} V                                                 & \cong \mat{(0 \mapsto m) \oplus (1 \mapsto p) \oplus (2 \mapsto w)                                                                                                                         & 0 & 0 \\ (0 \mapsto p) \oplus (1 \mapsto m) \oplus (2 \mapsto w) & m & 0 \\ 0 & 0 & m}= \mat{m \delta(0, 0) \oplus p \delta(1, 0) \oplus w \delta(2, 0) & 0 & 0 \\ p \delta(0, 0) \oplus m \delta(1, 0) \oplus w \delta(2, 0) & m & 0 \\ 0 & 0 & m} \\
		1 \xleftarrow{1} V'                                                & \cong \mat{(0 \mapsto m) \oplus (1 \mapsto p) \oplus (2 \mapsto w)                                                                                                                         & 0 & 0 \\ 0 & m & 0 \\ (0 \mapsto p) \oplus (1 \mapsto m) \oplus (2 \mapsto w) & 0 & m}
		= \mat{ m \delta(0, 1) \oplus p \delta(1, 1) \oplus w \delta(2, 1) & 0                                                                                                                                                                                          & 0     \\ 0 & m & 0 \\ p \delta(0, 1) \oplus m \delta(1, 1) \oplus w \delta(2, 1) & 0 & m}
	\end{aligned}\)

	Some care is needed to perform the addition for the I rule: the choices in the left and right branches are independent, so we must use coefficients in \(\{0, 1, 2\}^2 \to \mwp\) for the \(2^3\) choices.
	While the mapping notation would require to use positions to describe which choice is being refereed to, the \(\delta\) notation makes it immediate, as it encodes in the second value of \(\delta\) that two choices are considered, numbering the choice in the left branch \(0\).
	Hence we can sum the coefficients and obtain the matrix presented in the implementation as \href{\hrefeximproved}{\texttt{example7.c}}.
\end{example}

\begin{example}
	Our deterministic system now assigns to \pr|loop X3 {X2 = X1 + X2}| from \autoref{ex:start} the unique matrix
	\[
		\mat{m & (0 \mapsto p) \oplus (1 \mapsto m) \oplus (2 \mapsto w) &0 \\ 0 & (0 \mapsto m) \oplus (1 \mapsto \infty) \oplus (2 \mapsto \infty) & 0 \\ 0 & (0 \mapsto p) \oplus (1 \mapsto 0) \oplus (2 \mapsto 0) &m}
		=
		\mat{m & p \delta(0, 0) \oplus m \delta(1, 0) \oplus w \delta(2, 0) &0 \\ 0 & m \delta(0, 0) \oplus \infty \delta(1, 0) \oplus \infty \delta(2, 0) & 0 \\ 0 & p \delta (0, 0) \oplus 0 \delta (1,0) \oplus 0 \delta (2, 0) &m}
	\]
	where we observe that
	\begin{enumerate*}\item only one choice, one assignment, 0, gives a matrix without \(\infty\) coefficient, corresponding to the fact that, in the original system, only \(\pi_0\) could be used to complete the proof,
		\item the choice impacts the matrix locally, the coefficients being mostly the same, independently from the choice,
		\item the influence of \pr|X2| on itself is where possible non-polynomial growth rates lies, as the \(\infty\) coefficient are in the second column, second row.
	\end{enumerate*}

\end{example}

We are now in possession of all the material and intuitions needed to state the correspondence between our system and the original one of Jones and Kristiansen.

\begin{restatable}[Determinancy and termination]{theorem}{detterm}
	\label{lem:det-term}
	Given a program \(P\), there exists unique \(p \in \mathbb{N}\) and \(M \in \mathbb{M}(\{0,1,2\}^p\rightarrow\mwp^{\infty})\) such that \(\vdash P : M\).
\end{restatable}

\begin{proof}
	The existence of the matrix is guaranteed by the completeness of the
	rules, as any program written in the syntax presented in
	\autoref{subsec:language} can be typed with the rules of
	\autoref{fig:new-rules}.
	The uniqueness of the matrix is given by the fact that no two rules can be applied to the same command.
	Details are provided in \autoref{app:omitted-proofs}.
\end{proof}

\begin{restatable}[Adequacy]{theorem}{equiexpr}
	\label{lem:equi-expr}
	If \(\vdash P:M\), then for all \(\vec{a}\in A^p\), \(\vdashJK P: M[\vec{a}]\) iff \(\infty \notin M[\vec{a}]\).
\end{restatable}

\begin{proof}
	The proof uses that \(P\) cannot be assigned a matrix in the
	original calculus iff the deterministic calculus introduce a
	\(\infty\) coefficient, and from the fact that both calculus
	coincide in all the other cases.
	Details are provided in \autoref{app:omitted-proofs}.
\end{proof}

\begin{corollary}[Soundness]
	\label{cor:main-result}
	If \(\vdash P: M\) and there exists \(\vec{a}\in A^p\) such that \(\infty \notin M[\vec{a}]\), then every value computed by \(P\) is bounded by a polynomial in the inputs.
\end{corollary}

\begin{proof}
	This is an immediate corollary of the original soundness theorem~\cite[Theorem~5.3]{Jones2009} and of \autoref{lem:equi-expr}.
\end{proof}

This proves that the two analyses coincide, when excluding \(\infty\), and that we can re-use the original proofs.
However, our alternative definition should be understood as an important improvement, as it enables a better proof-search strategy while optimizing the memory usage, and hence enables the implementation (\autoref{sec:implementation}).
It also lets the programmer gain more fine-grained feedback, and illustrates the flexibility of the analysis: the latter will also be demonstrated by the improvements we discuss in the next section.

\section{Extending and improving the analysis: functions and efficiency}
\label{composition}

To improve this analysis, one could try to extract tight bound, to certify it, or to port it to a compiler's intermediate representation.
Adding constant values is arguably immediate~\cite[p. 3]{Jones2009} but handling pointers, even if technically possible, would probably require significant work.
This illustrates at the same time the flexibility of the analysis, and the distance separating ICC-inspired techniques from their usage on actual programs.
We decided to narrow this gap along two axes: the first one consists of allowing function definitions and calls in our syntax.
It is arguably a small improvement, but illustrates nicely the compositionality of the analysis, and includes recursively defined functions.
The second extension intersects the theory and the implementation: it details how our semi-ring structure can be leveraged to maintain a tractable algorithm to compute costly operations on our matrices, and to separate the problem of deciding if a bound exists from computing its form.

\subsection{Leveraging compositionality to analyze function calls}

Thanks to its compositionality, this analysis can easily integrate functions and procedures, by re-using the matrix and choices of a program implementing the function called.
We begin by adding to the syntax the possibility of defining multiple functions and calling them:

\begin{definition}[Functions]
	\label{def:extended-lang}
	Letting \pr|R| (\resp \pr|f|) range over variables (\resp function names), we add \emph{function calls}\footnote{Function calls that discard the output---procedures---could also be dealt with easily, but are vacuous in our effect-free, in particular pointer-free, language} to the commands (\autoref{def:lang}) and allow \emph{function	declarations}:

	\begin{center}
		\pr|C| \(\coloneqq\)\pr|Xi = f(X1, $\hdots$, Xn)| \hfill \pr|F| \(\coloneqq\) \pr|f(X1, $\hdots$, Xn){C; return R}|
	\end{center}
	In a function declaration, \pr|f(X1, $\hdots$, Xn)| is called the \emph{header}, and the \emph{body} is simply \pr|C| (\ie, \pr|return R| is not part of the body).
	A \emph{program} is now a series of function declarations such that all the function calls refer to previously declared functions---we deal with recursive calls in \autoref{ssec:recursion}---and a \emph{chunk} is a series of commands.
\end{definition}

Now, given a function declaration computing \(f\), we can obtain the matrix \(M_f\) by analyzing the body of \(f\) as previously done.
It is then possible to store the assignments \(\vec{a}_0, \hdots, \vec{a}_k\), for which no \(\infty\) coefficients appear\footnote{Allowing \(\infty\) coefficients would not change the method described nor its results, but it does not seem relevant to allow calling functions that are not polynomially bounded.}, and to project the resulting matrices to only keep the vector at \pr|R| that provides quantitative information about all the possible dependencies of the output variable \pr|R| \wrt input values, possibly merging choices leading to the same result.
After this, we are left with a family \((M_f[\vec{a}_0])\vert_{\text{\pr|R|}}, \hdots, (M_f[\vec{a}_k])\vert_{\text{\pr|R|}}\) of vectors---as the syntax here is restricted to functions with a single output value, even if accommodating multiple return values would be dealt with the same way---that we can re-use when calling the function.

The analysis of the command calling \(f\) is then dealt with the F rule below: %

\begin{center}
	\begin{prooftree}[small]
		\infer0[F]{\vdash \text{\pr|Xi = F(X1, $\hdots$, Xn)|} : 1 \xleftarrow{\text{\pr|i|}} (((M_f[\vec{a}_0])\vert_{\text{\pr|R|}})\delta(0, c) \oplus \dots \oplus ((M_f[\vec{a}_k])\vert_{\text{\pr|R|}})\delta(k, c))}
	\end{prooftree}
\end{center}

This rule introduces a choice \(c\) over \(k\) possible matrices, and it is possible that \(k \neq 3\), but this is not an issue, since our semi-ring construction can accommodate any set of choice \(A\).

\begin{example}
	\label{ex:function-call}
	Consider the following two programs \(Q\) and \(P\):

	{\centering

	\begin{tabular}{l c l}
		\(Q = \)
		\begin{lstlisting}[language=C, backgroundcolor = \color{white}]
int f(X1, X2){
 while b do {X2=X1+X1};
 return X2;
}
 \end{lstlisting}
		 & \hspace{4em}
		 &
		\(P = \)
		\begin{lstlisting}[language=C, backgroundcolor = \color{white}]
int foo(X1, X2){
 X2=X1+X1;
 X1=f(X2, X2);
}
 \end{lstlisting}
	\end{tabular}

	}

	We first have \(\vdash \text{\pr|X2 = X1 + X1|}:V\) for \(V = \mat{m \hspace{.2em} & p \delta(0, 0) \oplus p \delta(1, 0) \oplus w \delta(2, 0) \\0 & 0}\),
	and since \(V^{\ast} = \mat{m \hspace{.2em} & p \delta(0, 0) \oplus p \delta(1, 0) \oplus w \delta(2, 0) \\0 & m}\), applying W\(^\infty\) gives \(\vdash \text{\pr|Q|} : \mat{m \hspace{.2em} & \infty \delta(0, 0) \oplus \infty \delta(1, 0) \oplus w \delta(2, 0) \\ 0 \hspace{.2em} & m}\).
	Noting that only one choice gives an \(\infty\)-free matrix, we can now carry on the analysis of \pr|P|:

	{
	\centering

	\begin{prooftree}[small]
		\hypo{}
		\ellipsis{}{\vdash \text{\pr|X2 = X1 + X1|} : V}
		\infer0[F]{\vdash \text{\pr|X1 = f(X2, X2)|} : 1 \xleftarrow{\text{\pr|1|}} (\mat{w \\ m} \delta(0, c))}
		\infer2[C]{\vdash P : V \otimes 1 \xleftarrow{\text{\pr|1|}} (\mat{w \\ m} \delta(0, c))}
	\end{prooftree}

	}

	In this particular case, the \(c\) choice can be discarded, since only one option is available.%
\end{example}

Now, to prove that the F rule faithfully extends the analysis (\autoref{thm:analysis-function}), \ie preserves \autoref{cor:main-result}, we prove that the analysis of the program \enquote{inlining} the function call---as defined below---is, up to some bureaucratic variable manipulation and ignoring some \(\infty\) coefficients, the same as the analysis resulting from using our rule.
Intuitively, this mechanism provides the expected result because the choices in the function \emph{do not} affect the program calling it, and because their sets of variables are disjoint---except for the return variable.

\begin{definition}[In-lining function calls]
	\label{def:in-line}
	Let \(P\) be a chunk containing a call to the function \(f\), and \(F\) be the function declaration computing the function \(f\).
	The \emph{context} \(P[\cdot]\), a chunk containing a slot \([\cdot]\), is obtained by replacing in \(P\) the function call \pr|Xi=f(X1, $\hdots$, Xn)|, with \pr|X'1=X1; $\dots$; X'n=Xn; $[\cdot]$ Xi=R|, for \pr|R|, \pr|X'1|, \(\hdots\), \pr|X'n| fresh variables added to the header containing the chunk.

	The chunk \(\tilde{F}\) is obtained from the body of \(F\) by renaming the input variables to \pr|X'1|, \(\hdots\), \pr|X'n|, and the variable returned by \(F\) to \pr|R|.
	The code \(P[F]\) is finally obtained by computing the chunk \(\tilde{F}\), and inserting it in place of the symbol \([\cdot]\) in \(P[\cdot]\).
\end{definition}

That \(P\) and \(P[F]\) have, at the end of their executions, the same values stored in the variables of \(P\) is straightforward in our imperative programming language.

\begin{example}
	\label{ex:composition}
	The in-lining of \(Q\) in \(P\) from \autoref{ex:function-call} would give the following chunck \(\tilde{Q}\) and context \(P[\cdot]\), \(P[Q]\) being obtained by replacing in the latter \([\cdot]\) with the former:

	{\centering

	\begin{tabular}{l l }
		\(\tilde{Q} =\) \begin{lstlisting}[language=C, backgroundcolor = \color{white}]
while b do {R=X'1+X'1};
 \end{lstlisting}
		\hfill &
		\(P[\cdot] = \) \begin{lstlisting}[language=C, backgroundcolor = \color{white}]
int foo(X1, X2, X'1, R){
	X2=X1+X1;
	X'1=X2;
	|$[\cdot]$|
	X1=R;
}
\end{lstlisting}
	\end{tabular}
	}

	The analysis of \pr|P| (excluding the function call) and \pr|Q| is implemented at \href{https://statycc.github.io/pymwp/demo/\#implementation_paper_example15_a.c}{\texttt{example15a.c}}, and of \pr|P[Q]| at \href{https://statycc.github.io/pymwp/demo/\#implementation_paper_example15_b.c}{\texttt{example15b.c}}: this latter diverges with \autoref{ex:function-call} only up to projection and \(\infty\)-coefficients that are removed by F but not when in-lining the function call.
\end{example}

Now, we need to prove that the matrices \(M(P)\)---obtained by analyzing \(P\) and using the F rule for \pr|Xi=f(X1, $\hdots$, Xn);|---and \(M(P[F])\)---obtained by analyzing the inlined \(P[F]\)---are the same.
However, to avoid conflict with the variables and to project the matrices on the relevant values, some bureaucracy is needed: we write \(\Pi_P(M(P[F]))\) (\resp \((1-\Pi_P)(M(P[F]))\)) the projection of \(M(P[F])\) onto the variables in (\resp \emph{not} in) \(P\).
Some non-deterministic choices may appear within the (modified) chunk \(\tilde{F}\) inside \(P[F]\), \ie
\begin{itemize}
	\item the coefficients of \(M(P)\) are elements of the semi-ring \(\prod_{i=1}^{p+1} A_i \rightarrow \mathbb{M}(\mwp)\), with one particular choice corresponding to the F rule---we write the corresponding index \(i_0\);
	\item the coefficients of \(M(P[F])\) are elements of the semi-ring
	      \(\prod_{i=1}^{p+k} B_i\rightarrow \mathbb{M}(\mwp)\), where
	      \(k\) choices are made within the chunk \(\tilde{F}\)---we write the
	      corresponding indexes \(j_1,j_2,\dots,j_k\) (note these are in fact
	      consecutive indexes).
\end{itemize}
We note \(\pi: \{1,\dots,p+k\} \rightarrow \{1,\dots,p+1\}\) the projection of
the choices in \(P[F]\) onto the corresponding choices in \(P\), \ie
\(
\pi(j)=\left\{
\begin{tabular}{ll}
	\(j\)     & \text{ if \(j<j_1\)}               \\
	\(i_0\)   & \text{ if \(j_1 \leqslant j<j_k\)} \\
	\(j-k+1\) & \text{ if \(j_k< j\)}
\end{tabular}
\right.
\).
We note that each matrix used as axiom in the function call corresponds to a
specific assignment on indexes \(j_1,\dots,j_k\). We write
\(\Psi: A_{i_0}\rightarrow \prod_{i=j_1}^{j_k} B_i\) the corresponding
injection, extended to
\(\bar{\Psi}: \prod_{i=1}^{p+1} A_{i}\rightarrow \prod_{i=0}^{p+k} B_i\)
straightforwardly.

\begin{theorem}
	\label{thm:analysis-function}
	For all \(\vec{a}\) in \(\prod_{i=1}^{p+1} A_i\),
	\( (M(P))[\vec{a}]=(1-\Pi_P)(M(P[F]))[\bar{\Psi}(\vec{a})] \), and for all \(\beta\) in \(\prod_{i=0}^{p+k} B_i\) not in the image of
	\(\bar{\Psi}\), \((1-\Pi_P)(M(P[F])[\beta])\)
	contains \(\infty\). \end{theorem}

\begin{proof}
	It is sufficient to prove it for the simplest chunk \(P\) containing only one command \pr|Xi = f(X1, $\hdots$, Xn)|.
	This comes from the compositional nature of the analysis, as a sequence of commands is assigned the product of the matrices of each individual command.
	Then, checking the theorem in this case is a straightforward, though tedious (due to keeping track of all indices),	computation.
\end{proof}

\subsection{Integrating recursive calls, the easy way}
\label{ssec:recursion}

The question of dealing with self-referential, or recursive, calls, naturally arises when extending to function calls.
It turns out that our approach makes such cases easy to handle.

A program implementing a function \pr|rec| calling itself cannot use the F rule presented above as is, since the result of the analysis of \pr|rec| is precisely what we are trying to establish.
However, if \pr|rec| takes two input variables \pr|X1| and \pr|X2| and its return value is assigned to a third variable \pr|X3|, then we already know that the vector at \(3\) will need to be replaced by the vector capturing the dependency between \pr|X1|, \pr|X2|, and the return variable of \pr|rec| (which we will take to be \pr|X3| in our example).
The solution consists in replacing the actual values in this vector by variables \(\alpha\), \(\beta\) ranging over values in \(\mwp^{\infty}\), terminating the analysis with those variables, and then to resolve the equation---which is easy given the small size of the \(\mwp^{\infty}\) semiring.

As an example\footnote{Where we use variables that are not parameters, following \autoref{int-var}, and where our recursive call does not terminate: we are focusing on growth rates and not on termination, and keep the example compact.
}, consider the following program and compute the corresponding matrix:
\begin{minipage}[t]{0.3\textwidth}
	{\small
		\begin{lstlisting}[language=C, backgroundcolor = \color{white}]
int rec(X1, X2){
 X1 = X1 + X2;
 X3 = rec(X1, X2);
 return X3;
}
 \end{lstlisting}
	}
\end{minipage}
\hfill
\begin{minipage}[t]{0.65\textwidth}
	{\small
		\begin{align*}
			  & \mat{m \delta(0, 0) \oplus p \delta(1, 0) \oplus w \delta(2, 0) \quad & 0 \quad & 0                                                                               \\ p \delta(0, 0) \oplus m \delta(1, 0) \oplus w \delta(2, 0) \quad & m \quad & 0 \\ 0 \quad & 0 \quad & m}
			\otimes
			1 \xleftarrow{3} \mat{\alpha                                                                                                                                          \\ \beta \\0}\\
			= & \mat{m \delta(0, 0) \oplus p \delta(1, 0) \oplus w \delta(2, 0) \quad & 0 \quad & \alpha m \delta(0, 0) \oplus \alpha p \delta(1, 0) \oplus \alpha w \delta(2, 0) \\ p \delta(0, 0) \oplus m \delta(1, 0) \oplus w \delta(2, 0) \quad & m \quad & \alpha p \delta(0, 0) \oplus \alpha m \delta(1, 0) \oplus \alpha w \delta(2, 0) \oplus \beta \\ 0 \quad & 0 \quad & 0 }
		\end{align*}
	}
\end{minipage}

Using the assignments 0, 1 and 2 gives \(\mat{m & 0 & \alpha m \\ p & m & \alpha p \oplus \beta \\ 0 & 0 & 0}\), \(\mat{p & 0 & \alpha p \\ m & m & \alpha m \oplus \beta \\ 0 & 0 & 0}\) and \(\mat{w & 0 & \alpha w \\ w & m & \alpha w \oplus \beta \\ 0 & 0 & 0} \), and since the third vector should be equal to \(\mat{\alpha \\ \beta \\0}\), this gives three systems of equations:
\begin{align*}
	\left\{\begin{array}{rcl} \alpha m & = & \alpha \\ \alpha p \oplus \beta & = & \beta \end{array}\right.
	 &  &
	\left\{\begin{array}{rcl} \alpha p & = & \alpha \\ \alpha m \oplus \beta & = & \beta \end{array}\right.
	 &  &
	\left\{\begin{array}{rcl} \alpha w & = & \alpha \\ \alpha w \oplus \beta & = & \beta \end{array}\right.
\end{align*}

The smaller solution to the first (\resp second, third) equational system is \(\{\alpha=m; \beta=p\}\) (\resp \(\{ \alpha= p; \beta=p\}\) , \(\{\alpha=w; \beta=w\}\)), and as a consequence, we find two meaningful solutions (all others being larger than those): \(\mat{m \\ p \\ 0}\) and \(\mat{w \\ w \\ 0}\).%

\subsection{Taking advantage of polynomial structure to compute efficiently}
\label{sec:advantage}

Ensuring that the analysis is tractable is an important part of our contribution.
For a program accepting \(n\) different derivations and having \(k\) different derivations that cannot be completed, the original flow calculus must run at most \(k + 1\) times to find \emph{one} derivation, while our analysis outputs the \(k + n\) different derivations in one run, and then sorts them---as discussed next---by listing all the evaluations and looking for \(\infty\) values.
In this task, the C rule, that allows building programs from commands, is obviously crucial and consists simply in multiplying two matrices: however, since we are internalizing the choices, those matrices contain a mixture of functions from choices to coefficients in \(\mwp^{\infty}\) and of coefficients in \(\mwp\).
Multiplying such matrices is more costly, but also essential: an 8-line program such as \href{\hrefexplosiondemo}{\texttt{explosion.c}} requires to multiply elements of its matrix 34,992 times\footnote{The need to optimize functions is made even more obvious when we discuss benchmarking in \autoref{sec:benchmarks}.}.
This forces to represent and manipulate the elements of \(\prod_{i=1}^{p} A_i \rightarrow \mathbb{M}(\mwp)\)---setting aside \(\infty\) coefficients for a moment---cleverly: simple comparison showed that the improved algorithm presented below made the analysis roughly \emph{five times} faster (\autoref{app:sec:comparison}).

As discussed in \autoref{notation-delta}, elements of this semi-ring are represented as \emph{polynomials} \wrt the generating set given by the functions \(\delta(i,j): \prod_{i=1}^{p} A_i \rightarrow {\mwp}\) defined by
\(\delta(i,j)(a_1,\dots,a_p)=m\) if \(a_j=i\) and \(\delta(i,j)(a_1,\dots,a_p)=0\) otherwise, \ie an element of \(\prod_{i=1}^{p} A_i \rightarrow {\mwp}\) is represented as a polynomial \(\sum_{i=1}^{n} \alpha_i \prod_{j=1}^{k_i}\delta(a_{i,j},b_{i,j})\) with \(\alpha_i \in {\mwp}\).

This basis has an important property: the \emph{monomials} \(\alpha_i \prod_{j=1}^{k_i}\delta(a_{i,j},b_{i,j})\) in a polynomial can be ordered so that the product with another monomial is ordered, \ie if \(\alpha \leqslant \beta\) and both \(\alpha \times \gamma\) and \(\beta \times \gamma\) are non-zero, then \(\alpha \times \gamma \leqslant \beta \times \gamma\).
This order is leveraged to obtain efficient algorithms, similar to what is done using Gröbner bases for computation of standard polynomials~\cite{vanderHoeven2019}.
For instance, the algorithm for \href{\hrefpolymulti}{multiplication of polynomials} uses this property to compute the product of an ordered polynomial \(P\) with \(\sum_{i=1}^{n} \alpha_i \prod_{j=1}^{k_i}\delta(a_{i,j},b_{i,j})\):
\begin{enumerate}
	\item compute the products \(P_i= P\times \alpha_i \prod_{j=1}^{k_i}\delta(a_{i,j},b_{i,j})\) for all \(i\);
	\item compare and order a list \(L\) of all the first elements of those polynomials;
	      \item\label{algostep} append the smallest element to the result and remove it from the corresponding \(P_i\);
	\item insert the (new) first element of \(P_i\) to the list \(L\) if it exists;
	\item if \(L\) is non-empty, go back to step \ref{algostep}.
\end{enumerate}

When adding or multiplying polynomials, which consist of monomials, we check if a monomial is contained or included by another, and exclude all redundant cases (\cf \href{\hrefcontains}{\texttt{contains}} or
\href{\hrefincludes}{\texttt{includes}}).
This is also done when inserting monomials. Thus we keep polynomials free of implementation choices that we would otherwise have to handle during evaluation.

\subsection{Deciding the existence of a bound faster thanks to delta graphs}
\label{sec:delta-graph}

Adopting the \(\prod_{i=1}^{p} A_i \rightarrow {\mwp}^{\infty}\) semi-ring allows completing all derivations simultaneously, but remains to determine if there exists an assignment \(\vec{a}\in \prod_{i=1}^{p} A_i\) \st the resulting matrix is \(\infty\)-free, to decide whenever a program accepts a polynomial bound: this is the \emph{evaluation} step.
Despite the optimizations detailed above that simplifies the task, this phase remains particularly costly, since the number of assignment grows exponentially \wrt the number of choice, which is linear in the number of variables.
While this step is necessary (in one form or another) if one wishes to produce the actual mwp matrices certifying polynomial bounds, we implemented a specific data structure to keep track of assignments resulting in \(\infty\) coefficients on the fly, thus allowing the analysis to provide a qualitative answer quickly.
This section details how those \emph{delta graphs} allow to immediately determines whenever a polynomial bound exists without having to compute the corresponding matrix, something that was not possible in the original, non-deterministic, calculus.

A delta graph is a graph whose vertices are monomials. The graph is populated during the analysis by adding those monomials that appear with an infinite coefficient---\ie possible choices leading to \(\infty\) in the resulting matrix.
This graph is structured in layers: each layer corresponds to the size of the monomials (the number of deltas) it contains.
The intuition is that a monomial---or rather a list of deltas \(\delta(\_,\_)\)--- defines a subset of the space \(\prod_{i=1}^{p} A_i\); the less deltas in the monomial, the greater the subspace represented\footnote{Our intuitions here come from the standard topological structure of spaces of infinite sequences, where such a monomial represents a \enquote{cylinder set}, \ie an element of the standard basis for open sets.}.
As we populate the delta graph, we create edges within a given layer to keep track of differences between monomials: we add an
edge labeled \(i\) between two monomials if and only if they differ only on one
delta \(\delta(\_,i)\) (\ie one is obtained from the other by replacing the
first index of \(\delta(\_,i)\)). This is used to implement a \href{\hrefdgfusion}{\texttt{fusion}} method on delta graphs, which simplifies the structure: as soon as a monomial \(m\) in layer \(n\) has
\(\mathrm{Card}(A_i)-1\) outgoing edges labelled \(i\), we can remove all these monomials and insert a shorter monomial in layer \(n-1\), obtained from \(m\) by simply removing \(\delta(\_,i)\).
This implements the fact that \(\sum_{k=0}^{\mathrm{Card}(A_i)-1} m \delta(k,j) = m\).

Now, remember the delta graph represents the subspace of assignments for which an \(\infty\) appears.
If at some point the delta graph is completely simplified (\ie \enquote{fusions} to the graph with a unique monomial consisting in an empty list of \(\delta(\_,\_)\)), it means the whole space of assignments is
represented and no mwp-bounds can be found.
On the contrary, if the analysis ends with a delta graph different from the completely simplified one, at least one assignment exists for which no infinite coefficients appear, and therefore at least one mwp-bound exists.
This allows to answer the question \enquote{Is there at least one mwp-bound?} \emph{without actually computing said bounds}.
Based on the information collected in the delta graph and the matrix with polynomial coefficients, one can however recover all possible matrix assignments by going through all possible valuations.

This last part is \href{https://github.com/statycc/pymwp/blob/a39fe9a8cefa4be5a93380d66d1cb8162bb0ed01/pymwp/delta_iter.py}{implemented} with a specific iterator that leverages the information collected in the delta graph to skip large sets of valuations in a single step.
For instance, suppose the monomial \(\delta(1,1)\) lies in the delta graph---\ie that an infinite coefficient will be reached if the second index is equal to \(1\).
When asked the valuation after \((0,0,2,2)\) (and supposing that \(\mathrm{Card}(A_i)=3\) for all \(i\)), our \href{https://github.com/statycc/pymwp/blob/a39fe9a8cefa4be5a93380d66d1cb8162bb0ed01/pymwp/delta_iter.py#L36}{\texttt{delta\_iterator}} will jump directly to \((0,2,0,0)\), skipping all intermediate valuation of the form \((0,1,a,b)\) in a single step.
Similarly, it will jump from \((1,0,2,2)\) to \((1,2,0,0)\), again skipping several valuations at a time, providing a faster analysis.
Note that the implementation required care, to correctly jump when given additional informations from the delta graph, \eg to produce \((2,0,1,0)\) as the successor of \((0,0,2,2)\) if \(\delta(0,0)\), \(\delta(1,1)\) and \(\delta(0,2)\) all belong to the delta graph.

\section{Implementing, testing and comparing the analysis}
\label{sec:implementation}

Demonstrating the implementability of the improved and extended mwp-bounds analysis requires an implementation.
Our open-source solution, packaged through \href{https://pypi.org/project/pymwp/}{Python Package Index (PyPI)} as \href{https://pypi.org/project/pymwp/}{\texttt{pymwp}}, is a standalone command line tool, written in \texttt{Python}, that automatically performs growth-rate analysis on programs written in a \href{\hreffeatlist}{subset} of the \texttt{C} programming language.
For programs that pass the analysis, it produces a matrix corresponding to the input program and a list of valid derivation choices; and for programs that do not have polynomial bounds, it reports infinity. Our motivation for choosing \texttt{C} as the language of analysis resulted from its central role and similarity with the original \texttt{while} language. \texttt{Python} was an ideal choice for the implementation because of its plasticity, collection of libraries, and because it allowed partial reuse of a previous flow analysis tool~\cite{lqicm,Moyen2017,Moyen2017b}. The source code is available on \href{https://github.com/statycc/pymwp/}{Github}, along with an \href{https://statycc.github.io/pymwp/demo/}{online demo}, and detailed \href{https://statycc.github.io/pymwp/}{documentation}~\cite{pymwp_doc} describing its current supported features and functionality.
We now discuss how we tested and assesed it, and how it compares (or, rather \emph{does not} compare) to other similar approaches.

\subsection{Experimental evaluation}
\label{sec:benchmarks}

We allocated extensive focus and effort on \href{https://github.com/statycc/pymwp/tree/e59aeca6f690c5768adad360523525fb63a908ea/tests}{testing} and \href{https://github.com/statycc/pymwp/actions/workflows/profile.yaml}{profiling} our implementation, to ensure the correctness and efficiency of the analysis, and with the terminal objective of obtaining a usable tool. \href{https://github.com/statycc/pymwp/tree/main/c_files}{The test suite} includes 42 \texttt{C} programs, carefully designed to exercise different aspects of the analysis, ranging \href{https://github.com/statycc/pymwp/tree/8cf20f6157a8dfc3ce1593b766c31b1d49af4f77/c_files/basics}{from basic derivations}, to ones producing worst-case behavior (by yielding \eg \href{https://statycc.github.io/pymwp/demo/#other_dense.c}{dense matrices} or \href{https://statycc.github.io/pymwp/demo/#other_explosion.c}{exponential number of derivations}), and classical examples such as computing the \href{https://statycc.github.io/pymwp/demo/#other_gcd.c}{greatest common divisor} or \href{https://statycc.github.io/pymwp/demo/#infinite_exponent_1.c}{exponentiation}.

We refer to \href{https://github.com/statycc/pymwp/releases/tag/profile-latest}{our benchmarks} (presented in \autoref{sec:app:benchmarks}) for measured analysis results for each program.
The most salient aspect is that our analysis is extremely fast (the time is measured in \emph{milli}seconds) despite important numbers of function calls (in the 10k range, excluding builtin Python language calls, for 10-lines programs).
Even examples tailored to stress our implementation %
cannot make the analysis go over \emph{4 seconds}.
We cannot compare our implementation with implementations of the original analysis, since it has never been implemented, and (according to our attempts) cannot be implemented in any realistic manner.

\subsection{Related tools and incompatible metrics}
\label{ssec:lit-rev}

This work was inspired by the series of works of the flow analysis from the \enquote{Copenhagen school}~\cite{Ben-Amram2008,Jones1995}. %
The overall flow analysis approach is related in spirit to abstract interpretation~\cite{Cousot1977b,Cousot1977a}; that bounds \emph{transitions} between states (\eg commands) instead of states~\cite{Jones1995}.
This approach shaped the implementation of tools detecting loop quasi-invariants~\cite{Moyen2017,Moyen2017b}. %

Other communities share a similar goal of inferring resource-usage. Complexity analyzers such as SPEED~\cite{Gulwani2009} for \texttt{C++}, COSTA~\cite{Costa2007} for \texttt{Java} bytecode, ComplexityParser~\cite{Hainry2021} for \texttt{Java}, Resource Aware ML for \texttt{OCaml}~\cite{Lichtman2017} or Cerco~\cite{Cerco2014} and Verasco~\cite{Jourdan2015} for \texttt{C} generate (certified) cost or runtime analysis on (subsets of) imperative programming languages.
Embracing such a large diversity is difficult, but our technique is different from existing implementations and tools: most of them focus on worst-case resource-usage complexity or termination, while we are interested in upper-bounds on the final values of program variables, \ie we focus on \emph{growth} instead of actual values.
This makes the comparison with our approach difficult, but highlights at the same time its uniqueness in today's landscape of static analyzers.

Further, our approach provides other desirable properties:
\begin{enumerate*}
	\item it is compositional, which allows to \enquote{hot-plug} bounds of previously analyzed functions without additional work,
	\item it is modular, as the internal machinery can be altered---as in this paper---without having to re-develop the theory,
	\item it is language-independent, as it reasons abstractly on imperative languages, but can be applied to real programs, as our implementation illustrates, and should extend to more complex languages,
	\item it is lightweight and programmer-friendly, as it is fast, does not require annotations or to record value ranges,
	\item it studies growth independently from \eg iteration bounds, thus sidestepping difficult cases that worst-case analysis has to tackle, and
	\item it may enable tight bounds on programs, as it has been done recently~\cite{Ben-Amram2020} for a similar analysis~\cite{Ben-Amram2008}.
\end{enumerate*}
In particular compositionality is a highly desirable property--because otherwise the analysis needs to be re-run on programs or API whenever embedded into different pieces of software--yet difficult to achieve by most other approaches, as discussed and partially remedied recently~\cite{Carbonneaux2015}.
While we suppose one approach could be used to derive the result obtained by the other, we do believe the originality of our pioneering ICC-based approach may inspire new and original directions in static program analysis.

\section{Conclusion: limitations, strengths and future work}
\label{sec:conclusion}

This work attempts to illustrate the usefulness and applicability of ICC results, but also the need to refine and adapt them.
We showed that the mwp-flow analysis as originally described cannot scale to programs in a real programming language: while the considered analysis is definitely powerful and elegant, its mathematical nature let some costly operations go unchecked. However we have shown that, extended and coupled to optimizations techniques, its result allows the development of a novel and original static analysis technique on imperative programs, focused on \emph{growth} rather than on termination or worst-case bounds.

This work is a proof of concept and it has limitations, both theoretical and practical: the theory is missing memory uses, pointers, and arrays and the supported feature set of the implementation could be extended.
But instead of focusing on what this analysis \emph{cannot} perform, we would like to stress that all the tools are in place to perform similar analysis on intermediate representations of code in compilers, which will naturally simplify the task of fitting richer program syntax to our analysis, and brings this technique yet another step closer to practical use cases.

Our next steps include certifying the analysis using the Coq proof assistant~\cite{coqman}, and implementing the analysis in certified tools such as the Compcert compiler~\cite{Leroy2009} (or, more precisely, its static single assignment version~\cite{Barthe2014}) or certified-llvm~\cite{Zhao2013}.
The plasticity of both compilers and of the implemented analysis should facilitate porting our results and approaches to support further programming languages in addition to \texttt{C}.
As complexity analysis is notably difficult in Coq~\cite{Gueneau2019}, we believe a push in this direction would be welcome, and that ICC provides all the needed tools for it.

\clearpage

\makeatletter
\def\@pdfborderstyle{/S/U/W 0}%
\makeatother

\bibliographystyle{plainurl}
\bibliography{bib}

\begin{thebibliography}{10}

\bibitem{Costa2007}
Elvira Albert, Puri Arenas, Samir Genaim, German Puebla, and Damiano Zanardini.
\newblock {COSTA:} design and implementation of a cost and termination analyzer
  for java bytecode.
\newblock In Frank~S. de~Boer, Marcello~M. Bonsangue, Susanne Graf, and
  Willem~P. de~Roever, editors, {\em Formal Methods for Components and Objects,
  6th International Symposium, {FMCO} 2007, Amsterdam, The Netherlands, October
  24-26, 2007, Revised Lectures}, volume 5382 of {\em LNCS}, pages 113--132.
  Springer, 2007.
\newblock \href {https://doi.org/10.1007/978-3-540-92188-2_5}
  {\path{doi:10.1007/978-3-540-92188-2_5}}.

\bibitem{Cerco2014}
Roberto~M. Amadio, Nicholas Ayache, Fran{\c{c}}ois Bobot, Jaap Boender, Brian
  Campbell, Ilias Garnier, Antoine Madet, James McKinna, Dominic~P. Mulligan,
  Mauro Piccolo, Randy Pollack, Yann Régis{-}Gianas, Claudio~Sacerdoti Coen,
  Ian Stark, and Paolo Tranquilli.
\newblock Certified complexity (cerco).
\newblock In Ugo~Dal Lago and Ricardo Pe{\~{n}}a, editors, {\em Foundational
  and Practical Aspects of Resource Analysis - Third International Workshop,
  {FOPARA} 2013, Bertinoro, Italy, August 29-31, 2013, Revised Selected
  Papers}, volume 8552 of {\em LNCS}, pages 1--18. Springer, 2013.
\newblock \href {https://doi.org/10.1007/978-3-319-12466-7_1}
  {\path{doi:10.1007/978-3-319-12466-7_1}}.

\bibitem{Avanzini2017}
Martin Avanzini and Ugo~Dal Lago.
\newblock Automating sized-type inference for complexity analysis.
\newblock {\em Proc.\ ACM Program.\ Lang.}, 1({ICFP}):43:1--43:29, 2017.
\newblock \href {https://doi.org/10.1145/3110287} {\path{doi:10.1145/3110287}}.

\bibitem{Baillot2004}
Patrick Baillot and Kazushige Terui.
\newblock Light types for polynomial time computation in lambda-calculus.
\newblock In {\em LICS}, pages 266--275. IEEE Computer Society, 2004.
\newblock \href {https://doi.org/10.1109/LICS.2004.1319621}
  {\path{doi:10.1109/LICS.2004.1319621}}.

\bibitem{Barthe2014}
Gilles Barthe, Delphine Demange, and David Pichardie.
\newblock Formal verification of an {SSA}-based middle-end for compcert.
\newblock {\em ACM Trans.\ Program.\ Lang.\ Syst.}, 36(1):4:1--4:35, 2014.
\newblock \href {https://doi.org/10.1145/2579080} {\path{doi:10.1145/2579080}}.

\bibitem{Bellantoni1992}
Stephen~J. Bellantoni and Stephen~Arthur Cook.
\newblock A new recursion-theoretic characterization of the polytime functions
  (extended abstract).
\newblock In S.~Rao Kosaraju, Mike Fellows, Avi Wigderson, and John~A. Ellis,
  editors, {\em STOC}, pages 283--93. ACM, 1992.
\newblock \href {https://doi.org/10.1145/129712.129740}
  {\path{doi:10.1145/129712.129740}}.

\bibitem{BenAmram2010}
Amir~M. Ben{-}Amram.
\newblock On decidable growth-rate properties of imperative programs.
\newblock In Patrick Baillot, editor, {\em Proceedings International Workshop
  on Developments in Implicit Computational complExity, {DICE} 2010, Paphos,
  Cyprus, 27-28th March 2010}, volume~23 of {\em EPTCS}, pages 1--14, 2010.
\newblock \href {https://doi.org/10.4204/EPTCS.23.1}
  {\path{doi:10.4204/EPTCS.23.1}}.

\bibitem{Ben-Amram2020}
Amir~M. Ben{-}Amram and Geoff~W. Hamilton.
\newblock Tight polynomial worst-case bounds for loop programs.
\newblock {\em Log.\ Meth.\ Comput.\ Sci.}, 16(2), 2020.
\newblock \href {https://doi.org/10.23638/LMCS-16(2:4)2020}
  {\path{doi:10.23638/LMCS-16(2:4)2020}}.

\bibitem{Ben-Amram2008}
Amir~M. Ben{-}Amram, Neil~D. Jones, and Lars Kristiansen.
\newblock Linear, polynomial or exponential? complexity inference in polynomial
  time.
\newblock In Arnold Beckmann and Costas~Dimitracopoulos andBenedikt Löwe,
  editors, {\em Logic and Theory of Algorithms, 4th Conference on Computability
  inEurope, CiE 2008, Athens, Greece, June 15-20, 2008, Proceedings}, volume
  5028 of {\em LNCS}, pages 67--76. Springer, 2008.
\newblock \href {https://doi.org/10.1007/978-3-540-69407-6_7}
  {\path{doi:10.1007/978-3-540-69407-6_7}}.

\bibitem{Carbonneaux2015}
Quentin Carbonneaux, Jan Hoffmann, and Zhong Shao.
\newblock Compositional certified resource bounds.
\newblock In David Grove and Stephen~M. Blackburn, editors, {\em Proceedings of
  the 36th {ACM} {SIGPLAN} Conference on Programming Language Design and
  Implementation, Portland, OR, USA, June 15-17, 2015}, pages 467--478. ACM,
  2015.
\newblock \href {https://doi.org/10.1145/2737924.2737955}
  {\path{doi:10.1145/2737924.2737955}}.

\bibitem{Cousot1977b}
Patrick Cousot and Radhia Cousot.
\newblock Abstract interpretation: {A} unified lattice model for static
  analysis of programs by construction or approximation of fixpoints.
\newblock In Robert~M. Graham, Michael~A. Harrison, and Ravi Sethi, editors,
  {\em Conference Record of the Fourth {ACM} Symposium on Principles of
  Programming Languages, Los Angeles, California, USA, January 1977}, pages
  238--252. {ACM}, 1977.
\newblock URL: \url{http://dl.acm.org/citation.cfm?id=512950}, \href
  {https://doi.org/10.1145/512950.512973} {\path{doi:10.1145/512950.512973}}.

\bibitem{Cousot1977a}
Patrick Cousot and Radhia Cousot.
\newblock Static determination of dynamic properties of recursive procedures.
\newblock In Erich~J. Neuhold, editor, {\em Formal Description of Programming
  Concepts: Proceedings of the {IFIP} Working Conference on Formal Description
  of Programming Concepts, St. Andrews, NB, Canada, August 1-5, 1977}, pages
  237--278. North-Holland, 1977.

\bibitem{DalLago2012a}
Ugo Dal~Lago.
\newblock A short introduction to implicit computational complexity.
\newblock In Nick Bezhanishvili and Valentin Goranko, editors, {\em ESSLLI},
  volume 7388 of {\em LNCS}, pages 89--109. Springer, 2011.
\newblock \href {https://doi.org/10.1007/978-3-642-31485-8_3}
  {\path{doi:10.1007/978-3-642-31485-8_3}}.

\bibitem{Gulwani2009}
Sumit Gulwani, Krishna~K. Mehra, and Trishul Chilimbi.
\newblock Speed: Precise and efficient static estimation of program
  computational complexity.
\newblock In {\em Proceedings of the 36th Annual ACM SIGPLAN-SIGACT Symposium
  on Principles of Programming Languages}, POPL '09, page 127–139, New York,
  NY, USA, 2009. Association for Computing Machinery.
\newblock \href {https://doi.org/10.1145/1480881.1480898}
  {\path{doi:10.1145/1480881.1480898}}.

\bibitem{Gueneau2019}
Arma{\"{e}}l Guéneau.
\newblock {\em Mechanized Verification of the Correctness and Asymptotic
  Complexity of Programs. (Vérification mécanisée de la correction et
  complexité asymptotique de programmes)}.
\newblock PhD thesis, Inria, Paris, France, 2019.
\newblock URL: \url{https://tel.archives-ouvertes.fr/tel-02437532}.

\bibitem{Hainry2021}
Emmanuel Hainry, Emmanuel Jeandel, Romain P{\'{e}}choux, and Olivier Zeyen.
\newblock Complexityparser: An automatic tool for certifying poly-time
  complexity of java programs.
\newblock In Antonio Cerone and Peter~Csaba {\"{O}}lveczky, editors, {\em
  Theoretical Aspects of Computing - {ICTAC} 2021 - 18th International
  Colloquium, Virtual Event, Nur-Sultan, Kazakhstan, September 8-10, 2021,
  Proceedings}, volume 12819 of {\em LNCS}, pages 357--365. Springer, 2021.
\newblock \href {https://doi.org/10.1007/978-3-030-85315-0_20}
  {\path{doi:10.1007/978-3-030-85315-0_20}}.

\bibitem{Hoffmann2012b}
Jan Hoffmann, Klaus Aehlig, and Martin Hofmann.
\newblock Resource aware {ML}.
\newblock volume 7358 of {\em LNCS}, pages 781--786. Springer, 2012.
\newblock \href {https://doi.org/10.1007/978-3-642-31424-7\_64}
  {\path{doi:10.1007/978-3-642-31424-7\_64}}.

\bibitem{Jones2009}
Neil~D. Jones and Lars Kristiansen.
\newblock A flow calculus of \emph{mwp}-bounds for complexity analysis.
\newblock {\em ACM Trans.\ Comput.\ Log.}, 10(4):28:1--28:41, 2009.
\newblock \href {https://doi.org/10.1145/1555746.1555752}
  {\path{doi:10.1145/1555746.1555752}}.

\bibitem{Jones1995}
Neil~D. Jones and Flemming Nielson.
\newblock {\em Abstract Interpretation: A Semantics-Based Tool for Program
  Analysis}, volume~4 of {\em Handbook of Logic in Computer Science}, pages 527
  -- 636.
\newblock Oxford University Press, 1995.

\bibitem{Jourdan2015}
Jacques{-}Henri Jourdan, Vincent Laporte, Sandrine Blazy, Xavier Leroy, and
  David Pichardie.
\newblock A formally-verified {C} static analyzer.
\newblock In Sriram K.~Rajamani andDavid Walker, editor, {\em Proceedings of
  the 42nd Annual {ACM} {SIGPLAN-SIGACT} Symposium onPrinciples of Programming
  Languages, {POPL} 2015, Mumbai, India, January 15-17, 2015}, pages 247--259.
  {ACM}, 2015.
\newblock \href {https://doi.org/10.1145/2676726.2676966}
  {\path{doi:10.1145/2676726.2676966}}.

\bibitem{Lafont2004}
Yves Lafont.
\newblock Soft linear logic and polynomial time.
\newblock {\em Theor.\ Comput.\ Sci.}, 318(1):163--180, 2004.
\newblock \href {https://doi.org/10.1016/j.tcs.2003.10.018}
  {\path{doi:10.1016/j.tcs.2003.10.018}}.

\bibitem{Leivant1993}
Daniel Leivant.
\newblock Stratified functional programs and computational complexity.
\newblock In Mary~S. Van~Deusen and Bernard Lang, editors, {\em Conference
  Record of the Twentieth Annual ACM SIGPLAN-SIGACT Symposium on Principles of
  Programming Languages}, pages 325--333. {ACM} Press, 1993.
\newblock \href {https://doi.org/10.1145/158511.158659}
  {\path{doi:10.1145/158511.158659}}.

\bibitem{Leroy2009}
Xavier Leroy.
\newblock Formal verification of a realistic compiler.
\newblock {\em Commun.\ ACM}, 52(7):107--115, 2009.
\newblock \href {https://doi.org/10.1145/1538788.1538814}
  {\path{doi:10.1145/1538788.1538814}}.

\bibitem{Lichtman2017}
Benjamin Lichtman and Jan Hoffmann.
\newblock Arrays and references in resource aware {ML}.
\newblock In Dale Miller, editor, {\em 2nd International Conference on Formal
  Structures for Computation and Deduction, {FSCD} 2017, September 3-9, 2017,
  Oxford, {UK}}, volume~84 of {\em LIPIcs}, pages 26:1--26:20. Schloss
  Dagstuhl, 2017.
\newblock URL: \url{http://www.dagstuhl.de/dagpub/978-3-95977-047-7}, \href
  {https://doi.org/10.4230/LIPIcs.FSCD.2017.26}
  {\path{doi:10.4230/LIPIcs.FSCD.2017.26}}.

\bibitem{lqicm}
Lqicm on c toy parser.
\newblock URL: \url{https://github.com/statycc/LQICM_On_C_Toy_Parser}.

\bibitem{Moyen2017c}
Jean{-}Yves Moyen.
\newblock {\em Implicit Complexity in Theory and Practice}.
\newblock PhD thesis, University of Copenhagen, 2017.
\newblock URL: \url{https://lipn.univ-paris13.fr/~moyen/papiers/
  Habilitation_JY_Moyen.pdf}.

\bibitem{Moyen2017}
Jean{-}Yves Moyen, Thomas Rubiano, and Thomas Seiller.
\newblock Loop quasi-invariant chunk detection.
\newblock In Deepak D'Souza and K.~Narayan Kumar, editors, {\em {ATVA}}, volume
  10482 of {\em LNCS}. Springer, 2017.
\newblock \href {https://doi.org/10.1007/978-3-319-68167-2_7}
  {\path{doi:10.1007/978-3-319-68167-2_7}}.

\bibitem{Moyen2017b}
Jean{-}Yves Moyen, Thomas Rubiano, and Thomas Seiller.
\newblock Loop quasi-invariant chunk motion by peeling with statement
  composition.
\newblock In Guillaume Bonfante and Georg Moser, editors, {\em Proceedings 8th
  Workshop on Developments in Implicit Computational Complexity and 5th
  Workshop on Foundational and Practical Aspects of Resource Analysis,
  DICE-FOPARA@ETAPS 2017, Uppsala, Sweden, April 22-23, 2017}, volume 248 of
  {\em EPTCS}, pages 47--59, 2017.
\newblock URL: \url{http://arxiv.org/abs/1704.05169}, \href
  {https://doi.org/10.4204/EPTCS.248.9} {\path{doi:10.4204/EPTCS.248.9}}.

\bibitem{pymwp_doc}
pymwp's documentation, 2021.
\newblock URL: \url{https://statycc.github.io/pymwp/}.

\bibitem{coqman}
Coq Team.
\newblock Coq documentation, 2022.
\newblock URL: \url{https://coq.github.io/doc/}.

\bibitem{vanderHoeven2019}
Joris van~der Hoeven and Robin Larrieu.
\newblock Fast gr{\"o}bner basis computation and polynomial reduction for
  generic bivariate ideals.
\newblock {\em Applicable Algebra in Engineering, Communication and Computing},
  30(6):509--539, Dec 2019.
\newblock \href {https://doi.org/10.1007/s00200-019-00389-9}
  {\path{doi:10.1007/s00200-019-00389-9}}.

\bibitem{Zhao2013}
Jianzhou Zhao, Santosh Nagarakatte, Milo M.~K. Martin, and Steve Zdancewic.
\newblock Formal verification of {SSA}-based optimizations for {LLVM}.
\newblock In Hans{-}Juergen Boehm and Cormac Flanagan, editors, {\em {ACM}
  {SIGPLAN} Conference on Programming Language Design and Implementation,
  {PLDI} '13, Seattle, WA, USA, June 16-19, 2013}, pages 175--186. {ACM}, 2013.
\newblock \href {https://doi.org/10.1145/2491956.2462164}
  {\path{doi:10.1145/2491956.2462164}}.

\end{thebibliography}

\appendix
\makeatletter
\def\@pdfborder{0 0 0}%
\def\@pdfborderstyle{/S/U/W 1}%
\makeatother

\section{Technical appendix on semi-rings}
\label{app:sec:semi-ring}

\subsection{The mwp semi-ring}
\label{sec:app:mwp}

This subsection briefly recall the definition of semi-ring (\autoref{def:semi-ring}) and proves that the mwp semi-ring (\autoref{def:mwp-matrix-alg}) is indeed a semi-ring (\autoref{lem:mwp-is-a-semiring}).

\begin{definition}[Semi-ring]
	\label{def:semi-ring}
	A semi-ring \(\mathbb{S} =(S,0,1,+,\times )\) is specified by a set \(S\) and two binary operations \(+\) (addition) and \(\times \) (multiplication) such that \(\{0,1\} \in S\) and
	\begin{enumerate}
		\item \((S, 0,+)\) is a commutative monoid: the operation \(+\) is associative, commutative, and has \(0\) as the identity element,
		\item \((S, 1,\times )\) is a monoid: the operation \(\times \) is associative and has \(1\) as the identity element,
		\item the operation \(\times \) distributes with respect to \(+\): for all \(a, b, c \in S\), \(a \times (b + c) = a \times b + a \times c\) and \((b + c) \times a = b \times a + c \times a\)
		\item[]	We call \(\mathbb{S} \) a \emph{strong} semi-ring if, additionally, \emph{\(0\) annihilates \(S\)}, i.e.\
		\item \(0\times a=a\times 0=0\) for all \(a \in S\).
	\end{enumerate}
\end{definition}

\begin{lemma}[mwp semi-ring]
	\label{lem:mwp-is-a-semiring}
	The tuple \((\{0,m,w,p\}, 0, m, +, \times )\), with
	\begin{itemize}
		\item \( 0 < m < w < p \),
		\item \(\alpha + \beta = \begin{dcases*}
			      \alpha & if \(\alpha \geqslant \beta \) \\
			      \beta  & otherwise
		      \end{dcases*}\)
		\item \(\alpha \times \beta =
		      \begin{dcases*}
			      \alpha + \beta & if \(\alpha \neq 0\) and \(\beta \neq 0\) \\
			      0              & otherwise
		      \end{dcases*}\)
	\end{itemize}
	is a \emph{strong} semi-ring.
\end{lemma}

\begin{proof}
	We prove that \((\{0,m,w,p\}, 0, m, +, \times )\) as defined respects the conditions of \autoref{def:semi-ring}.
	The proof is straightforward but detailed nevertheless.
	\begin{description}
		\item[\((\{0,m,w,p\}, 0, +)\) is a commutative monoid]
			We prove that \((\{0,m,w,p\}, +)\) is a commutative monoid by showing that it is associative, commutative, and has \(0\) as identity.

			\begin{description}
				\item[Associativity]
					\((\alpha + \beta ) + \gamma = \alpha + (\beta + \gamma)\)
					\begin{description}
						\item[Case 1: \(\alpha \geqslant \beta \geqslant \gamma\)]
							\begin{align*}
								         &  & \alpha                     & = \alpha                    \\
								\implies &  & \alpha + \gamma            & = \alpha + \beta            \\
								\implies &  & (\alpha + \beta ) + \gamma & = \alpha + (\beta + \gamma)
							\end{align*}

						\item[Case 2: \(\alpha \geqslant \gamma \geqslant \beta\)]
							\begin{align*}
								         &  & \alpha                     & = \alpha                    \\
								\implies &  & \alpha + \gamma            & = \alpha + \gamma           \\
								\implies &  & (\alpha + \beta ) + \gamma & = \alpha + (\beta + \gamma)
							\end{align*}

						\item[Case 3: \(\beta \geqslant \alpha \geqslant \gamma\)]
							\begin{align*}
								         &  & \beta                      & = \beta                     \\
								\implies &  & \beta + \gamma             & = \alpha + \beta            \\
								\implies &  & (\alpha + \beta ) + \gamma & = \alpha + (\beta + \gamma)
							\end{align*}

						\item[Case 4: \(\beta \geqslant \gamma \geqslant \alpha\)]
							\begin{align*}
								         &  & \beta                      & = \beta                     \\
								\implies &  & \beta + \gamma             & = \alpha + \beta            \\
								\implies &  & (\alpha + \beta ) + \gamma & = \alpha + (\beta + \gamma)
							\end{align*}

						\item[Case 5: \(\gamma \geqslant \alpha \geqslant \beta\)]
							\begin{align*}
								         &  & \gamma                     & = \gamma                    \\
								\implies &  & \alpha + \gamma            & = \alpha + \gamma           \\
								\implies &  & (\alpha + \beta ) + \gamma & = \alpha + (\beta + \gamma)
							\end{align*}

						\item[Case 6: \(\gamma \geqslant \beta \geqslant \alpha\)]
							\begin{align*}
								         &  & \gamma                     & = \gamma                    \\
								\implies &  & \beta + \gamma             & = \alpha + \gamma           \\
								\implies &  & (\alpha + \beta ) + \gamma & = \alpha + (\beta + \gamma)
							\end{align*}
					\end{description}

				\item[Commutative Property]
					Both cases are immediate:
					\begin{description}
						\item[Case 1: \(\alpha \geqslant \beta\)] \(\implies \alpha + \beta = \alpha = \beta + \alpha\)
						\item[Case 2: \(\beta \geqslant \alpha\)] \(\implies \alpha + \beta = \beta = \beta + \alpha\)
					\end{description}

				\item[Identity element is \(0\)]
					\[0 + 0 = 0 \qquad 0 + m = m \qquad 0 + w = w \qquad 0 + p = p\]
			\end{description}

		\item[\((\{0,m,w,p\}, m, \times )\) is a monoid]
			We now prove that \((\{0,m,w,p\}, m, \times )\) is a monoid by showing that it is associative, has \(m\) as identity, and has \(0\) as the annihilator.

			\begin{description}
				\item [Associativity]

				      \((\alpha \times \beta ) \times \gamma = \alpha \times (\beta \times \gamma )\)

				      \begin{description}
					      \item[Case 1: ] \(\alpha, \beta, \gamma \in \{m,w,p\} \)

						      \(\alpha \times \beta = \alpha + \beta \) Associativity of operation + is shown in the proof of the commutative monoid, \((\{0,m,w,p \}, +)\).

					      \item[Case 2: ] \(\alpha\), \(\beta\), or \(\gamma\) equals \(0\)

						      By definition of multiplication, the product is \(0\).
				      \end{description}

				\item[Identity element is \(m\)]
					\begin{alignat*}{4}
						 & 0 \times m &  & = 0 &  & = m \times 0 \\
						 & m \times m &  & = m &  & = m \times m \\
						 & w \times m &  & = w &  & = m \times w \\
						 & p \times m &  & = p &  & = m \times p
					\end{alignat*}

				\item[0 annihilates \(\{0,m,w,p\}\)]
					\begin{alignat*}{4}
						 & 0 \times 0 &  & = 0 &  & = 0 \times 0 \\
						 & m \times 0 &  & = 0 &  & = 0 \times m \\
						 & w \times 0 &  & = 0 &  & = 0 \times w \\
						 & p \times 0 &  & = 0 &  & = 0 \times p
					\end{alignat*}
			\end{description}

		\item [Distribution of multiplication over addition]
		      We conclude by proving that \(\times \) distributes over \(+\).

		\item[Right Distribution]
			\(\alpha \times (\beta + \gamma) = (\alpha \times \beta) + (\alpha \times \gamma)\)
			\begin{description}
				\item[Case 1: \(\beta \geqslant \gamma\)]
					\begin{align*}
						\implies &  & \alpha \times \beta            & = \alpha \times \beta                              \\
						\implies &  & \alpha \times (\beta + \gamma) & = (\alpha \times \beta ) + (\alpha \times \gamma )
					\end{align*}

				\item[Case 2: \(\gamma \geqslant \beta\)]
					\begin{align*}
						\implies &  & \alpha \times \gamma           & = \alpha \times \gamma                             \\
						\implies &  & \alpha \times (\beta + \gamma) & = (\alpha \times \beta ) + (\alpha \times \gamma )
					\end{align*}
			\end{description}

		\item[Left Distribution]
			\( (\alpha + \beta ) \times \gamma = (\alpha \times \gamma) + (\beta \times \gamma)\)
			\begin{description}
				\item[Case 1: \(\alpha \geqslant \beta\)]
					\begin{align*}
						\implies &  & \alpha \times \gamma            & = \alpha \times \gamma                             \\
						\implies &  & (\alpha + \beta ) \times \gamma & = (\alpha \times \gamma ) + (\beta \times \gamma )
					\end{align*}

				\item[Case 3: \(\beta \geqslant \alpha\)]
					\begin{align*}
						\implies &  & \beta \times \gamma             & = \beta \times \gamma                                       \\
						\implies &  & (\alpha + \beta ) \times \gamma & = (\alpha \times \gamma ) + (\beta \times \gamma ) \qedhere
					\end{align*}
			\end{description}

	\end{description}
\end{proof}

\subsection{Matrix semi-ring}
\label{sec:app:matrix}

This subsection explains and details how matrices with coefficients in a semi-ring can be used to construct semi-rings.

\begin{lemma}[Matrix semi-ring]
	\label{lem:matrices}
	Given a strong semi-ring \(\mathbb{S} = (S, 0, 1, +, \times )\), the tuple \(\mathbb{M} = (M,\mathbf{0} ,\mathsfbf{1} ,\oplus ,\otimes )\), with
	\begin{itemize}
		\item \(M\) the set of all \(n \times n\) matrices over \(S\), for all \(n \in \mathbb{N}\),
		\item \(\mathbf{0} \) defined by \(M = \mathbf{0} \) iff \(M_{ij} = 0\) for all \(i\) and \(j\),
		\item \(\mathsfbf{1} \) defined by \(M = \mathsfbf{1} \) iff \(M_{ij} = 1\) for \(i = j\), \(M_{ij} = 0\) otherwise,
		\item \(\oplus \) defined by \(C = A \oplus B\) iff \(C_{ij} = A_{ij} + B_{ij}\),
		\item \(\otimes \) defined by \(C = A \otimes B\) iff \(C_{ij} = \sum_{k=1}^{n} A_{ik} \times B_{kj}\),
	\end{itemize}
	is a strong semi-ring.
\end{lemma}

\begin{proof}
	We prove that \(\mathbb{M} = (M,\mathbf{0} ,\mathsfbf{1} ,\oplus ,\otimes )\) as defined respects the conditions of \autoref{def:semi-ring}.
	Let \(A\),\(B\),\(C\) be \(n \times n\) matrices over \(S\) where \(n \in \mathbb{N}\).
	\begin{description}
		\item[\( (M,\mathbf{0} ,\mathsfbf{1} ,\oplus )\) is a commutative monoid]
			We prove that \((M, \oplus )\) is a commutative monoid by showing that it is associative, commutative, and has \(\mathbf{0} \) as identity.

			\begin{description}
				\item[Associativity]
					\((A \oplus B) \oplus C = A \oplus (B \oplus C)\) iff \( ((A \oplus B) \oplus C)_{ij} = (A \oplus (B \oplus C))_{ij}\) for all \(i\), \(j\).
					\begin{align*}
						((A \oplus B) \oplus C)_{ij} & = (A \oplus B)_{ij} + C_{ij}                                 \\
						                             & = (A_{ij} + B_{ij}) + C_{ij}                                 \\
						                             & = A_{ij} + (B_{ij} + C_{ij}) \tag{by associativity of \(+\)} \\
						                             & = A_{ij} + (B \oplus C)_{ij}                                 \\
						                             & = (A \oplus ( B \oplus C))_{ij}
					\end{align*}

				\item[Commutative Property]
					\(A \oplus B = B \oplus A\) iff \((A \oplus B)_{ij} = (B \oplus A)_{ij}\) for all \(i\), \(j\).
					\begin{align*}
						(A \oplus B)_{ij} & = A_{ij} + B_{ij}                                 \\
						                  & = B_{ij} + A_{ij} \tag{by commutativity of \(+\)} \\
						                  & = (B \oplus A)_{ij}
					\end{align*}

				\item[Identity element is \(\mathbf{0} \)]
					Let \(A = \mathbf{0} \), then \(A_{ij} = 0\) for all \(i\), \(j\), and \(\mathbf{0} \) is the identity element iff \(A_{ij} + B_{ij} = B_{ij}\) for all \(i\), \(j\)
					\begin{align*}
						(A \oplus B)_{ij} & = A_{ij} + B_{ij}                       \\
						                  & = 0 + B_{ij} \tag{by identity of \(+\)} \\
						                  & = B_{ij}
					\end{align*}
			\end{description}

		\item[\((M, \mathsfbf{1} , \otimes )\) is a monoid]
			We now prove that \((M, \otimes )\) is a monoid by showing that it is associative and has \(\mathsfbf{1} \) as identity.

			\begin{description}
				\item [Associativity]
				      \((A \otimes B) \otimes C = A \otimes (B \otimes C)\) iff \(((A \otimes B) \otimes C) _{ij} = (A \otimes (B \otimes C))_{ij}\) for all \(i\), \(j\).
				      \begin{align*}
					      ((A \otimes B) \otimes C) _{ij} & = (\sum_{k=1}^{n} A_{ik} \times B_{kj}) \otimes C                                                  \\
					                                      & = \sum_{l=1}^{n} (\sum_{k=1}^{n} A_{ik} \times B_{kj})_{il} \times C_{lj}                          \\
					                                      & = \sum_{l=1}^{n} \sum_{k=1}^{n} (A_{ik} \times B_{kl}) \times C_{lj}                               \\
					                                      & = \sum_{k=1}^{n} \sum_{l=1}^{n} A_{ik} \times (B_{kl} \times C_{lj})\tag{by assoc. of \(\times \)} \\
					                                      & = \sum_{k=1}^{n} A_{ik} \times (\sum_{l=1}^{n} B_{il} \times C_{lj})_{kj}                          \\
					                                      & = A \otimes (\sum_{l=1}^{n} B_{il} \times C_{lj})                                                  \\
					                                      & = (A \otimes ( B \otimes C))_{ij}
				      \end{align*}

				\item[Identity element is \(\mathsfbf{1} \)] \(A \otimes B = B \) and \(B \otimes A = B \) where \(A = \mathsfbf{1} \) iff \(A_{ij} = 1\) for \(i = j\) and \(A_{ij} = 0\) otherwise.
					\begin{align*}
						(A \otimes B)_{ij} & = \sum_{k=1}^{n} A_{ik} \times B_{kj}                                                                      \\
						                   & = (A_{ii} \times B_{ij}) + \sum_{\mathclap{k=1,k \neq i}}^{n} A_{ik} \times B_{kj}                         \\
						                   & = (1 \times B_{ij}) + \sum_{\mathclap{k=1,k \neq i}}^{n} 0 \times B_{kj}\tag{by def. of \(\mathsfbf{1} \)} \\
						                   & = (1 \times B_{ij}) + \sum_{\mathclap{k=1,k \neq i}}^{n} 0\tag{by annihilation prop. of \(0\)}             \\
						                   & = (1 \times B_{ij}) \tag{by identity of \(+\)}                                                             \\
						                   & = B_{ij} \tag{by identity of \(\times \)}
					\end{align*}

					\begin{align*}
						(B \otimes A)_{ij} & = \sum_{k=1}^{n} B_{ik} \times A_{kj}                                                                      \\
						                   & = (B_{ij} \times A_{jj}) + \sum_{\mathclap{k=1,k \neq j}}^{n} B_{ik} \times A_{kj}                         \\
						                   & = (B_{ij} \times 1) + \sum_{\mathclap{k=1,k \neq j}}^{n} B_{ik} \times 0\tag{by def. of \(\mathsfbf{1} \)} \\
						                   & = (B_{ij} \times 1) + \sum_{\mathclap{k=1,k \neq j}}^{n} 0\tag{by annihilation prop. of \(0\)}             \\
						                   & = (B_{ij} \times 1) \tag{by identity of \(+\)}                                                             \\
						                   & = B_{ij} \tag{by identity of \(\times \)}
					\end{align*}

				\item [\(\mathbf{0} \) annihilates \(M\)] \(A \otimes B = \mathbf{0} \) and \(B \otimes A = \mathbf{0} \) where \(A = \mathbf{0} \) iff \(A_{ij} = 0\) for all \(i\), \(j\).

				      \begin{align*}
					      (A \otimes B)_{ij} & = \sum_{k=1}^{n} A_{ik} \times B_{kj}                            \\
					                         & = \sum_{k=1}^{n} 0 \times B_{kj}\tag{by def. of \(\mathbf{0} \)} \\
					                         & = \sum_{k=1}^{n} 0 \tag{by annihilation prop. of \(0\)}          \\
					                         & = 0
				      \end{align*}

				      \begin{align*}
					      (B \otimes A)_{ij} & = \sum_{k=1}^{n} B_{ik} \times A_{kj}                             \\
					                         & = \sum_{k=1}^{n} B_{kj} \times 0 \tag{by def. of \(\mathbf{0} \)} \\
					                         & = \sum_{k=1}^{n} 0 \tag{by annihilation prop. of \(0\)}           \\
					                         & = 0
				      \end{align*}

				\item [Distribution of multiplication over addition]
				\item[Right Distribution]
					\(A \otimes (B \oplus C) = (A \otimes B) \oplus (A \otimes C)\) iff \((A \otimes (B \oplus C))_{ij} = ((A \otimes B) \oplus (A \otimes C))_{ij}\) for all \(i\), \(j\).

					\begin{align*}
						A \otimes (B \oplus C))_{ij} & = \sum_{k=1}^{n} \big( A_{ik} \times (B_{kj} + C_{kj})\big)                                                            \\
						                             & = \sum_{k=1}^{n} \big((A_{ik} \times B_{kj}) + ( A_{ik} \times C_{kj})\big) \tag{by right distribution of \(\times \)} \\
						                             & = \sum_{k=1}^{n} (A_{ik} \times B_{kj}) + \sum_{k=1}^{n} (A_{ik} \times C_{kj})                                        \\
						                             & = (A \otimes B)_{ij} + (A \otimes C)_{ij}                                                                              \\
						                             & = ((A \otimes B) \oplus (A \otimes C))_{ij}
					\end{align*}

				\item[Left Distribution]
					\((A \oplus B) \otimes C = (A \otimes C) \oplus ( B \otimes C)\) iff \(((A \oplus B) \otimes C)_{ij} = ((A \otimes C)\oplus ( B \otimes C))_{ij}\) for all \(i\), \(j\).

					\begin{align*}
						((A \oplus B) \otimes C)_{ij} & = \sum_{k=1}^{n} \big( (A_{ik} + B_{ik}) \times C_{kj} \big)                                                           \\
						                              & = \sum_{k=1}^{n} \big( (A_{ik} \times C_{kj}) + ( B_{ik} \times C_{kj})\big) \tag{by left distribution of \(\times \)} \\
						                              & = \sum_{k=1}^{n} (A_{ik} \times C_{kj}) + \sum_{k=1}^{n} (B_{ik} \times C_{kj})                                        \\
						                              & = (A \otimes C)_{ij} + (B \otimes C)_{ij}                                                                              \\
						                              & = ((A \otimes C) \oplus (B \otimes C))_{ij}
						\qedhere
					\end{align*}
			\end{description}
	\end{description}
\end{proof}

For simplicity, we will write \(\mathbb{M}\) as \(\mathbb{M} (\mathbb{S} ) = (M(S),\mathbf{0} ,\mathsfbf{1} ,\oplus ,\otimes )\).

\subsection{Choices semi-ring}
\label{sec:app:choice}

This subsection explains and details how functions into semi-ring coefficients can be used to construct semi-rings (\autoref{lem:functions}), and the interplay between this construction and the matrix semi-ring from the previous subsection (\autoref{lem:semi-ring-iso}) using the notion of semi-ring isomorphism (\autoref{def:iso}).

\begin{lemma}[Choices semi-ring]
	\label{lem:functions}
	Given a strong semi-ring \(\mathbb{S} = (S, 0, 1, +, \times )\) and a set \(A\), the tuple \(\mathbb{F} = (F, \mathsf{0} , \mathsf{1} , \boxplus , \boxtimes )\), with
	\begin{itemize}
		\item \(F\) the set of functions from \(A\) to \(S\),
		\item \(\mathsf{0} \) the constant function \(\mathsf{0} (a) = 0\) for all \(a \in A\),
		\item \(\mathsf{1} \) the constant function \(\mathsf{1} (a) = 1\) for all \(a \in A\),
		\item \(\boxplus \) defined componentwise: \((f \boxplus g)(a) = (f(a)) + (g(a))\), for all \(f\), \(g\) in \(F\) and \(a \in A\),
		\item \(\boxtimes \) defined componentwise: \((f \boxtimes g)(a) = (f(a)) \times (g(a))\), for all \(f\), \(g\) in \(F\) and \(a \in A\),
	\end{itemize}
	is a strong semi-ring.
\end{lemma}

\begin{proof}
	\begin{description}
		\item[ \((F, \mathsf{0} , \boxplus )\) is a commutative monoid] We first prove that \((F, \mathsf{0} , \boxplus )\) is a commutative monoid by showing that it is associative, commutative, and has \(\mathsf{0} \) as identity.
			\begin{description}
				\item[Associativity]
					\begin{align*}
						((f \boxplus g) \boxplus h)(a) & = (f(a) + g(a)) + h(a)                                          \\
						                               & = f(a) + (g(a) + h(a)) \tag{by assoc. of \(+\)}                 \\
						                               & = (f \boxplus (g \boxplus h))(a) \tag{by def. of \(\boxplus \)}
					\end{align*}

				\item[Commutativity]
					\begin{align*}
						(f \boxplus g)(a) & = f(a) + g(a)                                     \\
						                  & = g(a) + f(a)\tag{by commutativity of \(+\)}      \\
						                  & = (g \boxplus f)(a)\tag{by def. of \(\boxplus \)}
					\end{align*}

				\item[Identity element is \(0\)]
					\begin{align*}
						(\mathsf{0} \boxplus f)(a) & = \mathsf{0} (a) + f(a)                     \\
						                           & = 0 + f(a) \tag{by def. of \(\mathsf{0} \)} \\
						                           & = f(a) \tag{by identity prop of \(+\)}
					\end{align*}
			\end{description}

		\item[\((F, 1, \boxtimes )\) is a monoid] We now prove that \((F, 1, \boxtimes )\) is a monoid by showing that it is associative and has \(\mathsf{1} \) as identity.
			\begin{description}
				\item[Associativity]
					\begin{align*}
						((f \boxtimes g) \boxtimes h)(a) & = (f(a) \times g(a)) \times h(a)                                   \\
						                                 & = f(a) \times (g(a) \times h(a)) \tag{by assoc. of \(\times\) }    \\
						                                 & = (f \boxtimes (g \boxtimes h))(a) \tag{by def. of \(\boxtimes \)}
					\end{align*}

				\item[Identity element is \(1\)]
					\begin{align*}
						(\mathsf{1} \boxtimes f)(a) & = \mathsf{1} (a) \times f(a)                    \\
						                            & = 1 \times f(a)\tag{by def. of \(\mathsf{1} \)} \\
						                            & = f(a) \tag{by identity prop of \(\times\) }
					\end{align*}
			\end{description}

		\item [Distribution of multiplication over addition]
		      We conclude by proving that \(\boxtimes \) distributes over \(\boxplus \).
		      \begin{description}
			      \item[Right Distribution]
				      \begin{align*}
					      (f \boxtimes (g \boxplus h))(a) & = f(a) \times (g(a) + h(a))                                                          \\
					                                      & = (f(a) \times g(a)) + (f(a) \times h(a)) \tag{by right distribution of \(\times \)} \\
					                                      & = ((f \boxtimes g) \boxplus (f \boxtimes h))(a)
				      \end{align*}

			      \item[Left Distribution]
				      \begin{align*}
					      ((f \boxplus g) \boxtimes h)(a) & = (f(a) + g(a)) \times h(a)                                                        \\
					                                      & = (f(a) \times h(a)) + (g(a) \times h(a))\tag{by left distribution of \(\times \)} \\
					                                      & = ((f \boxtimes h) \boxplus (g \boxtimes h))(a)
				      \end{align*}
		      \end{description}

		\item[\(0\) annihilates \(F\)]
			\begin{align*}
				(\mathsf{0} \boxtimes f)(a) & = \mathsf{0} (a) \times f(a)                     \\
				                            & = 0 \times f(a) \tag{by def. of \(\mathsf{0} \)} \\
				                            & = 0 \tag{by annihilation prop of \(0\)}
			\end{align*}

			\begin{align*}
				(f \boxtimes \mathsf{0} )(a) & = f(a) \times \mathsf{0} (a)                     \\
				                             & = f(a) \times 0 \tag{by def. of \(\mathsf{0} \)} \\
				                             & = 0 \tag{by annihilation prop of \(0\)}
			\end{align*}

	\end{description}
\end{proof}

For simplicity, we will write \(\mathbb{F} \) as \(A \to \mathbb{S} = (A \to S, 0, 1, +, \times )\).

\begin{definition}[Semi-ring isomorphism]
	\label{def:iso}
	Two semi-rings \(\mathbb{S} =(S, 0, 1, +, \times )\) and \(\mathbb{T} = (T, \mathsf{0} , \mathsf{1} , \boxplus , \boxtimes )\) are \emph{isomorphic} and write \(\mathbb{S} \cong \mathbb{T} \) if there exists \(g: S \to T\) such that

	\begin{itemize}
		\item \(g\) is a bijection,
		\item \(g(0) = \mathsf{0} \),
		\item \(g(1) = \mathsf{1} \),
		\item \(g(s_1 + s_2) = g(s_1) \boxplus g(s_2)\) for all \(s_1, s_2 \in S\)
		\item \(g(s_1 \times s_2) = g(s_1) \boxtimes g(s_2)\) for all \(s_1, s_2 \in S\)
	\end{itemize}
	For simplicity, we write \(g: \mathbb{S} \to \mathbb{T} \) for such morphisms.
\end{definition}

\begin{lemma}
	\label{lem:semi-ring-iso}
	For all set \(A\) and strong semi-ring \(\mathbb{S} \), \(\mathbb{M} (A \to \mathbb{S} ) \cong A \to \mathbb{M} (\mathbb{S} )\).
\end{lemma}

\begin{proof}
	First, observe that by Lemmas~\ref{lem:matrices} and \ref{lem:functions}, both \(A \to \mathbb{M} (\mathbb{S} )\) and \(\mathbb{M} (A \to \mathbb{S} )\) are strong semi-rings, and we write \(0_f\) (resp.\ \(0_M\)) and \(1_f\) (resp.\ \(1_M\)) for the \(0\) and \(1\) elements of \(A \to \mathbb{M} (\mathbb{S} )\) (resp.\ of \(\mathbb{M} (A \to \mathbb{S} )\)).
	Now we have to prove that we can construct a bijection \(g : M(A \to S) \to (A \to M(S))\) that respects the conditions of \autoref{def:iso}.

	We define \(g\) and \(g^{-1}\) at the same time, then show that they are indeed inverses:
	\begin{description}
		\item[\(g : M(A \to S) \to (A \to M(S))\)]
			Given \(M \in M(A \to S)\) of size \(n \times n\), we let \(g(M) \in A \to M(S)\) be the function that maps \(a \in A\) to \(M\) where the same argument \(a\) has been applied to the functions \(f_{1,1}, \hdots, f_{n,n}\).
			Graphically:

			\[
				g(M)a =
				g(\begin{pmatrix}
					M_{1,1} & \ldots & M_{1,n} \\
					\vdots  & \ddots & \vdots  \\
					M_{n,1} & \ldots & M_{n,n}
				\end{pmatrix})a =
				\begin{pmatrix}
					M_{1,1}a & \ldots & M_{1,n}a \\
					\vdots   & \ddots & \vdots   \\
					M_{n,1}a & \ldots & M_{n,n}a
				\end{pmatrix}
			\]

			Below, we write \(f_M\) for \(g(M)\).

		\item[\(g^{-1}: (A \to M(S)) \to M(A \to S)\)]
			Given \(f \in A \to M(S)\), we define \(g^{-1}(f) \in M(A \to S)\) to be the matrix of size \(n \times n\), for \(n \times n\) the size of the matrix returned by \(f\), such that \((g^{-1}(f))_{i, j}\) is the function that maps \(a \in A\) to \((f(a))_{i,j}\) for all \(i\), \(j\).
			Graphically:

			\[
				g^{-1}(f)a =
				\begin{pmatrix}
					(fa)_{1,1} & \ldots & (fa)_{1,n} \\
					\vdots     & \ddots & \vdots     \\
					(fa)_{n,1} & \ldots & (fa)_{n,n}
				\end{pmatrix}
			\]

			Below, we write \(M_f\) for \(g^{-1}(f)\).
	\end{description}

	\begin{description}
		\item[\(g\) is a bijection]
			We first prove that \(g \circ g^{-1} = g^{-1} \circ g = \id\).
			\begin{description}
				\item[\((g^{-1} \circ g)(M) = M\)]
					\begin{align*}
						(g^{-1} \circ g)(M) & = g^{-1}(g(M))                                                  \\
						                    & = g^{-1}(f_M) \tag{\small{where \((f_M(a))_{ij} = M_{ij}(a)\)}} \\
						                    & = M
					\end{align*}

				\item[\((g \circ g^{-1})(f) = f\)]
					\begin{align*}
						(g \circ g^{-1})(f) & = g(g^{-1}(f))                                             \\
						                    & = g(M_f) \tag{\small{where \((M_f)_{ij}a = (f(a))_{ij}\)}} \\
						                    & = f
					\end{align*}
			\end{description}

		\item[\(g(0_M) = 0_f\)]
			Let \(f = g(0_M)\), then \(f = 0_f\) iff \(f(a)_{ij} = 0_{\mathbb{S}}\) for all \(i\), \(j\).
			\begin{align*}
				f(a)_{ij} & = (0_M)_{ij}(a)                           \\
				          & = 0_f(a)\tag{by def. of \(0_M\)}          \\
				          & = 0_{\mathbb{S}} \tag{by def. of \(0_f\)}
			\end{align*}

		\item[\(g(1_M) = 1_f\)]
			Let \(f = g(1_M)\), then \(f = 1_f\) iff \(f(a)_{ij} = 1_{\mathbb{S}} \) for all \(i=j\) and \(f(a)_{ij} = 0_{\mathbb{S}} \) otherwise.
			\begin{description}
				\item[Case 1: \(i = j\)]
					\begin{align*}
						f(a)_{ij} & = (1_M)_{ij}(a)                           \\
						          & = 1_f(a)\tag{by def. of \(1_M\)}          \\
						          & = 1_{\mathbb{S}} \tag{by def. of \(1_f\)}
					\end{align*}

				\item[Case 2: \(i \neq j\)]
					\begin{align*}
						f(a)_{ij} & = (1_M)_{ij}(a)                           \\
						          & = 0_f(a)\tag{by def. of \(1_M\)}          \\
						          & = 0_{\mathbb{S}} \tag{by def. of \(0_f\)}
					\end{align*}
			\end{description}

		\item[\(g(M_1 + M_2) = g(M_1) + g(M_2)\)]
			\begin{align*}
				     & g(M_1 + M_2) = g(M_1) + g(M_2)                                              \\
				\iff & f_{M_1 + M_2} = f_{M_1} + f_{M_2}                                           \\
				\iff & f_{M_1 + M_2}(a) = (f_{M_1} + f_{M_2})(a)                                   \\
				\iff & f_{M_1 + M_2}(a) = f_{M_1}(a) + f_{M_2} (a)                                 \\
				\iff & (f_{M_1 + M_2}(a))_{ij} = (f_{M_1}(a) + f_{M_2} (a))_{ij}                   \\
				\iff & (M_1 + M_2)_{ij}(a) = (M_1)_{ij}(a)+ (M_2)_{ij}(a) \tag{by assoc. of \(+\)}
			\end{align*}

		\item[\(g(M_1 \times M_2) = g(M_1) \times g(M_2)\)]
			\begin{align*}
				     & g(M_1 \times M_2) = g(M_1) \times g(M_2)                                                                                                      \\
				\iff & f_{M_1 \times M_2} = f_{M_1} \times f_{M_2}                                                                                                   \\
				\iff & f_{M_1 \times M_2}(a) = (f_{M_1} \times f_{M_2})(a)                                                                                           \\
				\iff & f_{M_1 \times M_2}(a) = (f_{M_1})(a) \times (f_{M_2})(a)                                                                                      \\
				\iff & (f_{M_1 \times M_2}(a))_{ij} = ((f_{M_1})(a) \times (f_{M_2})(a))_{ij}                                                                        \\
				\iff & (\sum_{k=1}^{n} (M_1)_{ik} \times (M_2)_{kj})(a) = \sum_{k=1}^{n} (M_1)_{ik}(a) \times (M_2)_{kj}(a) \tag{by assoc. of \(+\) and \(\times \)}
			\end{align*}
	\end{description}
\end{proof}

\subsection{Partiality semi-ring}
\label{sec:app:partiality}

In our improvement of the analysis, we add an \(\infty\) element to the mwp semi-ring, but reason abstractly below with an arbitrary semi-ring and a \(\bot\) element.

\begin{lemma}
	Given a strong semi-ring \(\mathbb{S} = (S,0,1,+,\times )\) and an element \(\bot \notin S\), \(\mathbb{S}^{\bot} = (S \cup \{\bot\},0,1,+^{\bot},\times^{\bot} )\) with, for all \(a\), \(b \in S \cup \{\bot\}\),
	\begin{align*}
		a +^{\bot} b      & =\begin{dcases*}
			                     a + b & if \(a, b \neq \bot\) \\
			                     \bot  & otherwise
		                     \end{dcases*}      \\
		a \times^{\bot} b & =\begin{dcases*}
			                     a \times b & if \(a, b \neq \bot\) \\
			                     \bot       & otherwise
		                     \end{dcases*}
	\end{align*}
	is a semi-ring.
\end{lemma}

\begin{proof}
	The proof is immediate, but note that \(\mathbb{S}^{\bot}\) is not strong, as \(\bot \times 0 = \bot\).
\end{proof}

A good intuition on this construction comes from partial functions.
Indeed, we can define \(A \rightharpoonup \mathbb{S}\) as the semi-ring of partial functions from \(A\) to \(\mathbb{S}\), i.e.\ of functions from \(A\) to \(\mathbb{S}^{\bot}\).
Furthermore, if we identify a matrix in \(\mathbb{M} (\mathbb{S}^{\bot})\) where at least a coefficient is \(\bot \) with the matrix \(\bot\), then we get that \(\mathbb{M} (A \rightharpoonup \mathbb{S} ) \cong A \rightharpoonup \mathbb{M} (\mathbb{S} )\).
However, note that none of those semi-rings are strong.

In the particular case of \(\mathbb{M}(\mwp^\infty)\), having \(0\times \infty=\infty\) instead of \(0\times \infty=0\) as required by the strength property allows to make sure that no non-polynomial growth is deleted.
Indeed, if part of the program computes an exponential value but then throws it away, having \(0\times \infty=0\) would hide the super-polynomial computation and results in an incorrect analysis.
However, \(0\times \infty=0\) could still be useful, at the cost of losing the bounds on time and space usage for terminating programs, but providing the benefit of analyzing programs that \emph{ultimately} have polynomial dependency of the values \wrt the inputs.

\section{Omitted Proofs}
\label{app:omitted-proofs}

\detterm*

\begin{proof}
	The proof proceeds by induction on the length of the program \(P\), expressed in number of commands.
	We let \(p\) be the number of variables in \(P\), but observe that any program \(P\) can be treated as manipulating \(p'>p\) different variables, by simply adding \(p'-p\) additional rows and columns to the matrix, and leaving them unchanged by the derivation of \(P\).
	While a complete proof would need to constantly account for the number of actual and potential variables used by \(P\), we will simply assume that the reader understands that accounting for this technicality obfuscate more than it clarifies the proof, and we will freely resize the matrices to account for additional variables when needed.
	\begin{description}
		\item[If \(P\) is of length \(1\)]
			Then we know \(P\) is of the form \pr|X = e|, and only the rule A can be applied.
			But then we need to prove that all expression \pr|e| can be typed with exactly one vector.
			An expression \pr|e| is either a variable \pr|X|, or a composed expression \pr|X * Y|, \pr|X - Y|, or \pr|X + Y|.
			But then, respectively, only E\(^{\textsc{S}}\), E\(^{\textsc{M}}\) or E\(^{\textsc{A}}\) (for addition and substraction) can be applied, and this case is proven.
		\item[If \(P\) is of length \(n > 1\)]
			The we proceed by case on the structure of the command:
			\begin{itemize}
				\item If \(P\) is of the form \pr|if b then P1 $\text{\hspace{.2em}}$ else P2|, then by induction we know for \(i \in \{1, 2\}\) there exists \(p_i\) and \(M_i\) of size \(p_i \times p_i\) such that \(\vdash \text{\pr|Pi|} : M_i\).
				      If \(p_1 \neq p_2\), then letting \pr|Mj| being the smaller matrix, it is easy to rewrite \pr|Pj|'s derivation to account for \(|p_1 - p_2|\) additional variables, and as \(\oplus\) is uniquely defined, we know that \(M_1 \oplus M_2\) results in a unique matrix of size \(\max(p_1, p_2)\).
				\item If \(P\) is of the form \pr|while b do P'|, this is immediate by induction hypothesis on \pr|P'|, considering that only W\(^{\infty}\) can be applied, and that this rule produces a unique matrix.
				\item If \(P\) is of the form \pr|loop X {P'}|, this case is similar to the previous one, using L\(^{\infty}\) instead of W\(^{\infty}\).
				\item If \(P\) is of the form \pr|P1;P2|, this case is similar to the \pr|if| case, with the possible need to resize one of the matrix obtained by induction, and using that \(\otimes\) is uniquely defined.\qedhere
			\end{itemize}
	\end{description}
\end{proof}

\equiexpr*

\begin{proof}
	The proof proceeds by induction on the length of the program \(P\), expressed in number of commands.
	\begin{description}
		\item[If \(P\) is of length \(1\)]
			Then we know \(P\) is of the form \pr|X = e|, and only the rule A can be applied, in both systems.
			Hence, we need to prove that all expression \pr|e| can be typed the same way in both systems.
			A careful comparison of Figures~\ref{fig:orig-rules} and \ref{fig:new-rules} shows that if \pr|e| is of the form \pr|Xi|, then there is a small mismatch.
			In the original system, we can use either E2, and obtain \(\vdashJK \text{\pr{Xi}} : \{ _{\text{\pr|i|}}^{w}\}\), or E1, and obtain \(\vdashJK \text{\pr{Xi}} : \{ _{\text{\pr|i|}}^{m}\}\), while the only derivation in the deterministic system is using E\(^{\textsc{S}}\) to get \(\vdashJK \text{\pr{Xi}} : \{ _{\text{\pr|i|}}^{m}\}\).
			As \(m < w\), we argue that the deterministic system cannot obtain a derivation that is not useful anyway, and hence that it can be ignored.

			As for the other cases, if \pr|e| is a composed expression \pr|X * Y|, \pr|X - Y|, or \pr|X + Y|, it is easy to observe that E\(^{\textsc{A}}\) and E\(^{\textsc{M}}\) encapsulates all the possible combinations of E2 and of E1 followed by E3 or E4 that can be used.
		\item[If \(P\) is of length \(n > 1\)]
			Then the result holds by induction, once we observed that L\(^{\infty}\) and W\(^{\infty}\) are introducing \(\infty\) coefficients \emph{only if} L and W cannot be applied.
			\qedhere
	\end{description}
\end{proof} %
\section{Benchmarks}
\label{sec:app:benchmarks}

\subsection{Descriptions of program groups}

\begin{itemize}
	\item \textit{Basics} -- C programs performing operations corresponding to simple derivation trees.
	\item \textit{Implementation paper} -- example programs presented in this paper.
	\item \textit{Original paper} -- examples taken from or inspired by the original analysis~\cite{Jones2009}.
	\item \textit{Infinite} -- programs whose matrices always contain infinite coefficients.
	\item \textit{Polynomial} -- programs whose matrices do not always contain infinite coefficients.
	\item \textit{Other} -- other C programs of interest.
\end{itemize}

\subsection{Results}

The benchmarks are categorized and grouped to distinguish the type of system behavior they exercise.
For each program we capture in \autoref{tab:table-name}
\begin{enumerate}
	\item program variable count
	\item the lines of code in the source program (LOC column)
	\item clock time taken by the full analysis (excluding saving result to file, which is otherwise default behavior),
	\item number of function calls excluding builtin Python language calls, and
	\item the result of the analysis.
\end{enumerate}

Collectively the LOC, time, and function calls columns provide insight into the behavior of the analysis as different aspects of the system are being stress-tested. From the results column we report expected results on each benchmarked program. In the benchmarks table a passing result is represented with \checkmark and $\infty$ otherwise. We do not report manually computed bounds as comparison, because the analysis is carried out on individual variables, thus calculating them on multivariate programs is tedious and futile. However, for simple programs such as \href{https://statycc.github.io/pymwp/demo/#basics_while_2.c}{while\_2.c}, it is straightforward through visual inspection to verify the obtained $2 \times 2$-matrix is indeed the correct result.

These benchmarks were obtained using Python's built-in cProfile utility, extended in \texttt{pymwp} implementation to enable batch profiling. The clock times are slight overestimates because the utility adds minor runtime overhead. The number of function calls includes primitive calls, but exclude built-in Python language calls.
Full detailed results are viewable in the source code repository: \href{https://github.com/statycc/pymwp/releases/tag/profile-latest}{https://github.com/statycc/pymwp/releases/tag/profile-latest}

Details of executing machine: linux (Ubuntu), OS release: 20.04.3 (LTS), version: 20220131.1; CPU Cores: 2; CPU model: Intel(R) Xeon(R) CPU E5-2673 v4 @ 2.30GHz; Kernel release: 5.11.0-1028-azure; total memory bytes: 7284846592; Python version: 3.10.2 (x64).

\subsection{Comparison}
\label{app:sec:comparison}

It is not really meaningful or possible to compare those results with any other static analyzer, and impossible to compare it with any other implementation of this type of flow analysis.
While we could, in theory, analyze our examples with other static analyzers, their results would be incomparable, as they would produce guarantees on termination or worst case resource usage, which are both orthogonal to our polynomial bounds on value growth.
To our knowledge, the only static analyzer using similar metrics~\cite{Avanzini2017} was developed only for functional languages, thus preventing comparison.
As for implementations of the original analysis, our first attempts showed that a naive implementation would likely fail to handle the memory or time explosions.
We did, however, compare the gains resulting from the optimizations described in \autoref{sec:advantage}.
In a nuthsell, our \href{https://github.com/statycc/pymwp/pull/18/commits/fc2068cdedf9560879294b71f009ee780cf3ca86}{improved algorithm} for \href{https://statycc.github.io/pymwp/polynomial/#pymwp.polynomial.Polynomial.add}{adding} and \href{https://statycc.github.io/pymwp/polynomial/#pymwp.polynomial.Polynomial.times}{multiplying} polynomials resulted in the analysis being \href{https://github.com/statycc/pymwp/issues/17#issuecomment-854931398}{roughly \emph{five times faster}} for two programs that we estimate to be representative.

\begin{table}
	\centering
	\begin{tabular}{|l|c|c|r|r|c|}
		\hline
		\textbf{Program name}         & \textbf{Variables} & \textbf{LOC} & \textbf{Time (ms)} & \textbf{Function calls} & \textbf{Bound} \\
		\hline
		\rowcolor{lipicsgrey}
		\textit{Basics}               &                    &              &                    &                         &                \\
		assign\_expression            & 2                  & 8            & 133                & 81614                   & \checkmark     \\
		assign\_variable              & 2                  & 9            & 115                & 81238                   & \checkmark     \\
		if                            & 2                  & 9            & 118                & 82046                   & \checkmark     \\
		if\_else                      & 2                  & 7            & 118                & 82928                   & \checkmark     \\
		inline\_variable              & 2                  & 9            & 118                & 81979                   & \checkmark     \\
		while\_1                      & 2                  & 7            & 117                & 82934                   & \checkmark     \\
		while\_2                      & 2                  & 7            & 117                & 83964                   & \checkmark     \\
		while\_if                     & 3                  & 9            & 122                & 91572                   & \checkmark     \\
		\hline
		\rowcolor{lipicsgrey}
		\textit{Implementation paper} &                    &              &                    &                         &                \\
		example7                      & 3                  & 10           & 122                & 86898                   & \checkmark     \\
		example15\_a                  & 2+2                & 25           & 122                & 88763                   & \checkmark     \\
		example15\_b                  & 4                  & 16           & 137                & 122016                  & \checkmark     \\
		\hline
		\rowcolor{lipicsgrey}
		\textit{Original paper}       &                    &              &                    &                         &                \\
		example3\_1\_a                & 3                  & 10           & 110                & 85286                   & \checkmark     \\
		example3\_1\_b                & 3                  & 10           & 120                & 87637                   & \checkmark     \\
		example3\_1\_c                & 3                  & 11           & 121                & 89173                   & \checkmark     \\
		example3\_1\_d                & 2                  & 12           & 116                & 80002                   & $\infty$       \\
		example3\_2                   & 3                  & 12           & 118                & 83182                   & $\infty$       \\
		example3\_4                   & 5                  & 18           & 134                & 108890                  & $\infty$       \\
		example5\_1                   & 2                  & 10           & 116                & 81185                   & \checkmark     \\
		example7\_10                  & 3                  & 10           & 119                & 86053                   & \checkmark     \\
		example7\_11                  & 4                  & 11           & 139                & 119379                  & \checkmark     \\
		\hline
		\rowcolor{lipicsgrey}
		\textit{Infinite}             &                    &              &                    &                         &                \\
		exponent\_1                   & 4                  & 16           & 127                & 99893                   & $\infty$       \\
		exponent\_2                   & 4                  & 13           & 123                & 92846                   & $\infty$       \\
		infinite\_2                   & 2                  & 6            & 143                & 128275                  & $\infty$       \\
		infinite\_3                   & 3                  & 9            & 120                & 89880                   & $\infty$       \\
		infinite\_4                   & 5                  & 9            & 3274               & 5924420                 & $\infty$       \\
		infinite\_5                   & 5                  & 11           & 369                & 529231                  & $\infty$       \\
		infinite\_6                   & 4                  & 14           & 1624               & 2836726                 & $\infty$       \\
		infinite\_7                   & 5                  & 15           & 631                & 964189                  & $\infty$       \\
		infinite\_8                   & 6                  & 23           & 880                & 1444782                 & $\infty$       \\
		\hline
		\rowcolor{lipicsgrey}
		\textit{Polynomial}           &                    &              &                    &                         &                \\
		notinfinite\_2                & 2                  & 4            & 119                & 86174                   & \checkmark     \\
		notinfinite\_3                & 4                  & 9            & 131                & 104826                  & \checkmark     \\
		notinfinite\_4                & 5                  & 11           & 169                & 168242                  & \checkmark     \\
		notinfinite\_5                & 4                  & 11           & 174                & 176179                  & \checkmark     \\
		notinfinite\_6                & 4                  & 16           & 195                & 215765                  & \checkmark     \\
		notinfinite\_7                & 5                  & 15           & 1161               & 1961806                 & \checkmark     \\
		notinfinite\_8                & 6                  & 22           & 1893               & 3172293                 & \checkmark     \\
		\hline
		\rowcolor{lipicsgrey}
		\textit{Other}                &                    &              &                    &                         &                \\
		dense                         & 3                  & 16           & 157                & 151428                  & \checkmark     \\
		dense\_loop                   & 3                  & 17           & 269                & 353068                  & \checkmark     \\
		explosion                     & 18                 & 23           & 1296               & 2327071                 & \checkmark     \\
		gcd                           & 2                  & 12           & 114                & 84914                   & $\infty$       \\
		simplified\_dense             & 2                  & 9            & 118                & 85098                   & \checkmark     \\
		\hline
	\end{tabular}
	\caption{ Benchmark results produced by \texttt{pymwp} on \texttt{C} programs.}
	\label{tab:table-name}
\end{table}

\end{document}